\documentclass[10pt,twocolumn,twoside]{IEEEtran}
\usepackage[utf8]{inputenc}
\usepackage{amsmath}
\usepackage{mathrsfs}

\usepackage{amssymb}
\usepackage{braket}
\usepackage{amsmath,mathtools}
\usepackage{float}
\usepackage{graphicx}
\usepackage{enumerate}
\usepackage{bbm}
\DeclareMathOperator{\Tr}{Tr}
\usepackage{verbatim}
\usepackage{accents}
\usepackage{authblk}
\usepackage{amsmath, amssymb, bbm, xspace}
\usepackage{epsfig}
\usepackage{longtable}
\usepackage{color}
\usepackage{mathrsfs}
\usepackage{comment}
\usepackage{ifthen}
\newboolean{showcomments}
\setboolean{showcomments}{true}
\usepackage{courier}
\usepackage{xcolor}
\usepackage{todonotes}
\usepackage[framemethod=tikz]{mdframed}
\usepackage{lineno}
\newmdenv[leftmargin=\dimexpr-0.4em, innerleftmargin=0.5em,
rightmargin=\dimexpr-0.4em, innerrightmargin=0.5em,
linewidth=2pt,linecolor=red, topline=false, bottomline=false,
innertopmargin=0pt,innerbottommargin=0pt,skipbelow=0pt,skipabove=0pt,%
]{notex}

\newenvironment{note}%
{\vskip\dimexpr\dp\strutbox-\prevdepth\relax\notex\strut\ignorespaces}%
{\xdef\notetpd{\the\prevdepth}\endnotex\vskip-\notetpd\relax}

\let\oldtodo\todo

\makeatletter%
\DeclareDocumentCommand{\todo}{ O{} +g +d<> }{%
		\setlength{\marginparwidth}{1.5cm}%
	\IfNoValueTF{#2}{\relax}{%
		\oldtodo[caption={#2},size=\scriptsize,#1]{\renewcommand{\baselinestretch}{0.8}\selectfont\sffamily#2\par}%
	}%
	\IfNoValueTF{#3}{\relax}{%
		\IfNoValueTF{#2}{% when parmark but no todo
			\begin{note}%
				\begin{internallinenumbers}%
					\indent%
					#3%
				\end{internallinenumbers}%
			\end{note}%
		}{% when parmark and todo
			\vspace{-0\baselineskip}%
			\begin{note}%
				\begin{internallinenumbers}%
					\indent%
					#3%
				\end{internallinenumbers}%
			\end{note}%
		}%
	}%
}%
\makeatother
\usepackage{soul}

\newcommand{\hlc}[2][yellow]{{%
		\colorlet{foo}{#1}%
		\sethlcolor{foo}\hl{#2}}%
}

\newcommand{\removetodo}[2]{\todo[color=pink]{\textbf{delete:} ``#1'' #2}\hlc[pink]{#1}}
\newcommand{\inserttodo}[1]{\todo[color=green!40]{\textbf{insert:} #1}}

\newcommand{\hltodoy}[2]{\todo[color=yellow!40]{#2}\hl{#1} }

\newcommand{\hltodoc}[3]{\todo[color=#3!40]{#2}\hlc[#3]{#1} }

\newcommand{\hltodo}[2]{\todo[color=orange!40]{#2}\hlc[orange!40]{#1} }
\newcommand{\replacetodo}[2]{\todo[color=pink!40]{\textbf{replace with:}``#2'' }\hl{#1} }

\usepackage{marginnote}
\newcommand{\todol}[1]{{%
		\let\marginpar\marginnote
		\reversemarginpar
		\renewcommand{\baselinestretch}{0.8}%
		\todo{#1}}}
	
\newcommand{\inserttodol}[1]{{%
		\let\marginpar\marginnote
		\reversemarginpar
		\renewcommand{\baselinestretch}{0.8}%
		\inserttodo{#1}}}
	
\newcommand{\removetodol}[2]{{%
		\let\marginpar\marginnote
		\reversemarginpar
		\renewcommand{\baselinestretch}{0.8}%
		\removetodo{#1}{#2}}}

\newcommand{\hltodol}[2]{{%
		\let\marginpar\marginnote
		\reversemarginpar
		\renewcommand{\baselinestretch}{0.8}%
		\hltodo{#1}{#2}}}

\newcommand{\replacetodol}[2]{{%
		\let\marginpar\marginnote
		\reversemarginpar
		\renewcommand{\baselinestretch}{0.8}%
		\replacetodo{#1}{#2}}}

\newcommand{\hltodoyl}[2]{{%
		\let\marginpar\marginnote
		\reversemarginpar
		\renewcommand{\baselinestretch}{0.8}%
		\hltodoy{#1}{#2}}}

\newcommand{\hltodocl}[3]{{		\let\marginpar\marginnote
		\reversemarginpar
		\renewcommand{\baselinestretch}{0.8}%
		\hltodoc{#1}{#2}{#3}}}

\newtheorem{theorem}{Theorem}[section]

\newtheorem{lemma}[theorem]{Lemma}

\newtheorem{proposition}[theorem]{Proposition}

\newtheorem{corollary}[theorem]{Corollary}

\newtheorem{definition}{Definition}[section]

\newcommand{\comb}[2]{\prescript{#1}{}C_{#2}}

\def\bkE{{\rm I\kern-.17em E}}
\def\bk1{{\rm 1\kern-.17em l}}
\def\bkD{{\rm I\kern-.17em D}}
\def\bkR{{\rm I\kern-.17em R}}
\def\bkP{{\rm I\kern-.17em P}}

\def\bkZ{{\bf{Z}}}

\def\bkE{{\rm I\kern-.17em E}}
\def\bk1{{\rm 1\kern-.17em l}}
\def\bkD{{\rm I\kern-.17em D}}
\def\bkR{{\rm I\kern-.17em R}}
\def\bkP{{\rm I\kern-.17em P}}

\makeatletter
\newcommand{\pushright}[1]{\ifmeasuring@#1\else\omit\hfill$\displaystyle#1$\fi\ignorespaces}
\newcommand{\pushleft}[1]{\ifmeasuring@#1\else\omit$\displaystyle#1$\hfill\fi\ignorespaces}
\makeatother

\def\bkZ{{\bf{Z}}}
\def\b12{(\beta_1,\beta_2)}

\newenvironment{example}{{\noindent \bf Example}}{\hfill $\square$\hspace{-4.5pt}\vspace{6pt}}
\newcounter{example}
\renewcommand{\theexample}{\thesection.\arabic{example}}

\newcounter{remark}
\renewcommand{\theremark}{\thesection.\arabic{remark}}

\def\t{^\top}
\def\Bscr{\mathscr{B}}

\def\Ebb{\mathbb{E}}
\newlength{\noteWidth}
\long\def\notes#1{\ifinner
{\tiny #1}
\else
\marginpar{\parbox[t]{\noteWidth}{\raggedright\tiny #1}}
\fi\typeout{#1}}

 \def\notes#1{\typeout{read notes: #1}} %uncomment for final version

\newcommand{\ie}{i.e.\@\xspace} %%% i.e.,
\newcommand{\eg}{e.g.\@\xspace} %%% e.g.,
\newcommand{\e}[2]{{\small e}$\scriptstyle#1$#2}
\newcommand{\Real}{\ensuremath{\mathbb{R}}}

\def\Ebb{\mathbb{E}}

\def\Pbb{{\mathbb{P}}}

\def\Ibf{{\bf I}}

\def\ast{\buildrel{t}\over\rightarrow}

\def\half  {{\textstyle{1\over 2}}}

\def\conv{\textrm{conv}\:}

\def\spose#1{\hbox to 0pt{#1\hss}}

\def\text #1{\hbox{\quad#1\quad}}

\def\nthinsp{\mskip -2   mu}

\def\superstar{^{\raise 0.5pt\hbox{$\nthinsp *$}}}
\def\SUPERSTAR{^{\raise 0.5pt\hbox{$*$}}}
\def\lamstarT {\lambda^{\raise 0.5pt\hbox{$\nthinsp *$}T}}

\def\Bscr{{\cal B}}

\def\Dscr{{\cal D}}

\def\Hscr{{\cal H}}
\def\Lscr{{\cal L}}

\def\Pscr{{\cal P}}
\def\Qscr{{\cal Q}}
\def\Sscr{{\cal S}}

\def\Uscr{{\cal U}}

\def\Nscr{{\cal N}}

\def\Gscr{{\cal G}}
\def\Cscr{{\cal C}}

\def\aur{\;\textrm{and}\;}

\def\non{\nonumber}

\def\bfI{{\bf I}}
\let\forallnew\forall
\renewcommand{\forall}{\forallnew\ }
\let\forall\forallnew

		\def\bkE{{\rm I\kern-.17em E}}
		\def\bk1{{\rm 1\kern-.17em l}}
		\def\bkD{{\rm I\kern-.17em D}}
		\def\bkR{{\rm I\kern-.17em R}}
		\def\bkP{{\rm I\kern-.17em P}}
		\def\bkY{{\bf \kern-.17em Y}}
		\def\bkZ{{\bf \kern-.17em Z}}
		\def\bkC{{\bf  \kern-.17em C}}
{\begin{list}{}%
         {\setlength{\leftmargin}{#1}}%
         \item[]%
}
{\end{list}}

\def\Tsf{\mathsf{T}}
\def\Esf{\mathsf{E}}

\def\Rsf{\mathsf{R}}
		\def\bsp{\begin{split}}
		\def\beq{\begin{eqnarray}}
		\def\bal{\begin{align*}}
		\def\bc{\begin{center}}
		\def\be{\begin{enumerate}}
		\def\bi{\begin{itemize}}
		\def\bs{\begin{small}}
		\def\bS{\begin{slide}}
		\def\ec{\end{center}}
		\def\ee{\end{enumerate}}
		\def\ei{\end{itemize}}
		\def\es{\end{small}}
		\def\eS{\end{slide}}
		\def\eeq{\end{eqnarray}}
		\def\eal{\end{align*}}
		\def\esp{\end{split}}
		\def\qed{ \vrule height7.5pt width7.5pt depth0pt}  %width4.17pt depth0pt} 

	\def\cp2problem#1#2#3#4{\fbox
		 {\begin{tabular*}{0.9\textwidth}
			{@{}l@{\extracolsep{\fill}}l@{\extracolsep{6pt}}l@{\extracolsep{\fill}}c@{}}
				#1 & & $#4 $ 
			\end{tabular*}}}

		\def\bkE{{\rm I\kern-.17em E}}
		\def\bk1{{\rm 1\kern-.17em l}}
		\def\bkD{{\rm I\kern-.17em D}}
		\def\bkR{{\rm I\kern-.17em R}}
		\def\bkP{{\rm I\kern-.17em P}}
		
		\def\bkZ{{\bf{Z}}}

\newcommand {\beeq}[1]{\begin{equation}\label{#1}}
\newcommand {\eeeq}{\end{equation}}
\newcommand {\bea}{\begin{eqnarray}}
\newcommand {\eea}{\end{eqnarray}}

\def\texitem#1{\par\smallskip\noindent\hangindent 25pt
               \hbox to 25pt {\hss #1 ~}\ignorespaces}

\def\bsp{\begin{split}}
		\def\beq{\begin{eqnarray}}
		\def\bal{\begin{align*}}
		\def\bc{\begin{center}}
		\def\be{\begin{enumerate}}
		\def\bi{\begin{itemize}}
		\def\bs{\begin{small}}
		\def\bS{\begin{slide}}
		\def\ec{\end{center}}
		\def\ee{\end{enumerate}}
		\def\ei{\end{itemize}}
		\def\es{\end{small}}
		\def\eS{\end{slide}}
		\def\eeq{\end{eqnarray}}
		\def\eal{\end{align*}}
		\def\esp{\end{split}}
		\def\qed{ \vrule height7.5pt width7.5pt depth0pt}  %width4.17pt depth0pt} 

\usepackage{amsmath, amssymb, xspace}
\usepackage{epsfig}
\usepackage{longtable}
\usepackage{color}
\usepackage{mathrsfs}
\usepackage{subfig}
\def\Cscr{{\cal C}}
\def\Nscr{{\cal N}}

\def\conv{{\rm conv}}
\reversemarginpar
\title{\bf The Quantum Advantage in Binary Teams and the Coordination Dilemma:  Part I}
\author[1]{Shashank A. Deshpande}
 \author[1]{Ankur A. Kulkarni} 
\affil[1]{Systems and Control Engineering,  Indian Institute of Technology Bombay, Mumbai 400076 \\ Emails: \texttt{shashank.deshpande.phy@iitb.ac.in, kulkarni.ankur@iitb.ac.in}}
\date{}

\begin{document}
\maketitle

\begin{abstract}
We have shown that entanglement assisted stochastic strategies allow access to strategic measures beyond the classically correlated measures accessible through passive common randomness, and thus attain a quantum advantage in decentralised control. In this two part series of articles, we investigate the decision theoretic origins of the quantum advantage within a broad superstructure of problem classes.  Each class in our binary team superstructure corresponds to a parametric family of cost functions with a distinct algebraic structure. In this part, identify the only problem classes that benefit from quantum strategies. We find that these cost structures admit a special decision-theoretic feature -- `\textit{ the coordination dilemma}'. Our analysis hence reveals some intuition towards the utility of non-local quantum correlations in decentralised control.
\end{abstract}

\section{Introduction}
Strategies correlated through passive common randomness constitute a well-known space of classically implementable strategies in decentralised control. However, it is known, thanks to a counter-example by Ananthram and Borkar~\cite{ananthram2007commonrandom} that this constitution does not exhaust the space of occupation measures allowed by the information structure of the problem. We have recently shown that this `limitation' of common randomness in decentralised control can be alleviated with the use of a quantum mechanical architecture~\cite{deshpande2022quantum} to generate randomness. Specifically, we considered a decentralised estimation problem and demonstrated a new class of entanglement assisted stochastic strategies that, while still respecting the information structure, produce a cost improvement over what is achievable through common randomness -- a phenomenon we called the \textit{quantum advantage} in decentralised control.

However, we also found, numerically, that the quantum advantage varies with our problem parameters and problem structure. It appears prima facie that the structure of the cost function, and the constants involved in it critically determine the manifestation of the quantum advantage. Moreover, the quantum advantage appears and disappears as we vary the fidelity of the observation channels of the decision makers. In this two part series of articles, our goal is to shed light on these occurrences from a decision-theoretic vantage point. Our aim is to delineate  decision theoretic features of the problem that explain or characterize these  observations.

In this paper,  we  define a superstructure of binary static team problems, and  situate the decentralised estimation problem  from~\cite{deshpande2022quantum} within  this superstructure. The problem from~\cite{deshpande2022quantum} demanded that the players coordinate on nonoverlapping subsets of actions for each value of the environmental uncertainty, as part of estimation effort. Since the environmental uncertainty is observed only partially, such a cost function imposes a \textit{coordination dilemma} on the decision makers. It is apparent from~\cite{deshpande2022quantum} that richer forms of correlations afforded by quantum entanglement are instrumental in extracting a quantum advantage in the face of this dilemma.
In this paper, we ask the converse question:  is it the case that problems that do not have such a dilemma, also do not benefit from quantum strategies? We find that the answer is \textit{yes} -- problems where the coordination dilemma exists are the \textit{only} class admitting a quantum advantage within our problem class superstructure. In all other classes, quantum strategies perform just as well as classical ones, \textit{regardless} of the parameter values. Thus decision-theoretically the coordination dilemma is in some sense fundamental to the appearance of a quantum advantage. This is the main contribution of the present paper. We find it to be a point of great caution that not every problem class admits a quantum advantage. In the second part, we investigate parametric values within this problem class that allow for the quantum advantage, thereby further constraining the instances where quantum strategies are useful. These parametric values capture the intensity of the coordination dilemma and the quality of information of the agents. More details about this are discussed in the second part.

Our quantum strategies require that the decision makers share a pair of  entangled particles; creation of such particles and protecting their state from decoherence is part of an ongoing technological effort. Some highly successful implementations exist, \eg, quantum key distribution~\cite{yin2020qcryp} has been realised by several distribution networks like the DARPA, SECOQC, SwissQuantum and Tokyo QKD (see also~\cite{gisin2007qcomm} and  \cite{brunner2014bell}). These successes notwithstanding, quantum strategies enabled through entanglement are, as of today, an expensive and fragile resource.
Our results give a sharp decision-theoretic boundary to ascertaining when such a resource is worth investing in.  Though we work with the specific class of binary problems, there are hints in our  calculations that our results can be used as ingredients in a larger structural investigation of more general problems. This more general analysis is part of our ongoing research.

Non-classical correlations arising from entanglement and their implications for the nature of physical reality have been a subject of the greatest of scientific debates, starting with Einstein~\cite{einstein1935epr} and Bohr~\cite{bohr1935epr}, later to Bell~\cite{bell1964epr} and more recently CHSH~\cite{clauser1969chsh} and Aspect~\cite{aspect1982violation}. The Nobel prize in physics in 2022 was awarded for the experimental confirmation of the existence of these very correlations. For stochastic control, the feature is the \textit{geometry} of nonclassical correlations (which reduce to occupation measures in stochastic control). The set of classically attainable distributions satisfy what are known as \textit{Bell inequalities}; these can be expressed as faces of the polytope formed by classical distributions. A violation of the Bell inequalities by a physical experiment implies these inequalities form a separating hyperplane between the distribution attained by the experiment and \textit{all} classically attainable distributions. A cost function can be loosely thought of as a `normal' to this hyperplane, which if aligned appropriately, exhibits a quantum advantage. However, this intuition is loose due to two reasons: first, it is not only the cost function but also the probability distribution of the observations and the environmental uncertainty that appears in the cost, and second, the environmental uncertainty is not observed by the players. More importantly, a geometric picture such as this does not allow for much understanding on the underlying decision-theoretic dilemmas that are at play. Our finding relating the coordination dilemma to the quantum advantage is thus also a novel, decision-theoretic insight into the powers of quantum correlations. Moreover, the finding is rather strong -- in the absence of the coordination dilemma, the quantum advantage does not manifest for any values of the parameter, \ie, the parameters in the cost or the probability distribution.

This article is organised as follows. In section \ref{sec:superstructure} we elaborate the coordination dilemma and develop the problem class superstructure that rest of the article investigates upon. In section \ref{sec:nonlocal} we briefly introduce different strategic classes and define the non-local advantages with respect to our superstructure. We also offer here, a fresh control theoretic perspective on some well known attributes of correlations in quantum information theory. Section \ref{sec:equivalences} investigates some invariances that our class superstructure enjoys, that compress our elimination procedure in Section \ref{sec:eliminate} to a few `generating' problem classes. Ultimately in Section \ref{sec:eliminate} we exhaustively scan through our superstructure and isolate problem classes that admit the quantum advantage.

\section{The Problem Class Superstructure and the Coordination Dilemma}\label{sec:superstructure}
\subsection{Notation}
We use $\Pscr(S)$ to denote the set of probability distributions on the set $S$. Similarly, $\Pscr(S|T)$ denotes the the conditional probability distribution on $S$ given an element in $T$. We denote by $\Bscr(\Hscr)$, the set of all complex bounded linear operators on a Hilbert space $\Hscr$. We employ the following notation for operations among boolean variables $a,b\in\{0, 1\}$.
$a\cdot b$ denotes the logical AND, $a\vee b$ denotes the logical OR, $a\oplus b$ denotes the logical XOR, $\sim a$ denotes the negation of $a$.  We denote the conjugate transpose of an operator $\rho\in\Bscr(\Hscr)$ by $\rho^\dagger$. We use $\Tr(\rho)$ to denote its trace. We denote $\mathbf{I}$ to be the identity matrix (operator); its ambient dimension would be clear from the context. Let $\Rsf $ denote the permutation matrix given by \begin{equation}\label{eq:rdef}
	\Rsf:=\begin{pmatrix}
		0&1\\
		1&0
	\end{pmatrix}
\end{equation}

\subsection{Team decision problems}
We consider decentralised decision problems with static information structure with two agents (or decision makers or players) $A$ and $B$. The state of nature is described as a tuple $(\xi_A, \xi_B, \xi_W)$ of correlated binary random variables with a known distribution $ \Pbb $ where $\xi_A \in \Xi_A$, $\xi_B \in \Xi_B$, $ \xi_W \in \Xi_W$  and $\Xi_A$ $=\Xi_B$ $=\Xi_W$ $=\{0, 1\}$.
Players $ A $ and $ B $ observe $ \xi_A $ and $ \xi_B $ respectively and must choose actions $ u_A  \in \Uscr_A$, and $ u_B \in \Uscr_B, $ respectively as a function of their observations.
Action spaces of $A$ and $B$ are sets $\mathcal{U}_A:=\{u_A^{0}, u_A^1\}$ and $\mathcal{U}_B:=\{u_B^0, u_B^1\}$ respectively; note that $ u_A $, $ u_B $ (without the superscript) denote generic elements of $ \Uscr_A,\Uscr_B $, respectively. Occasionally we will need to order the elements of $\Uscr_A$ and $\Uscr_B$, in which case it will be convenient to think of $\Uscr_A,\Uscr_B$ as two dimensional vectors (say in $\Real^2$) with distinct components each.

Based on their actions $u_A\in \Uscr_A$, $u_B \in \Uscr_B $ and the value of $\xi_W$, the decision makers incur a cost $\ell(u_A, u_B, \xi_W)$. The goal of the players is to  minimize $ \Ebb[\ell(u_A,u_B,\xi_W)] $, where the expectation is taken with respect to $ u_A,u_B $, $ \xi_A,\xi_B $ and $ \xi_W $, via a suitable choice of a policy. A policy is conditional joint distribution of $ u_A,u_B $ given $ \xi_A,\xi_B $,
$ Q(\cdot|\cdot) \in \Pscr(\Uscr|\Xi) $,
where $ \Uscr:= \Uscr_A \times \Uscr_B $ and $ \Xi := \Xi_A \times \Xi_B $. In addition to belonging to $ \Pscr(\Uscr|\Xi) $, a policy must also satisfy some constraints, the nature of which form a central topic in this paper; we defer this discussion to the next section.

\subsection{Coordination dilemma}
The cost function $ \ell $ depends on $ \xi_{W} $ which is not observed by either player, although players do observe $ \xi_{W} $ partially through $ \xi_A,\xi_B $.  The lack of knowledge of $ \xi_W $ is a source of a significant dilemma for the players. As an illustration consider the following cost function $ \ell(u_A,u_B,\xi_W) $  from~\cite{deshpande2022quantum}.
\begin{equation}
    \begin{array}{|c|c|c|}
    \hline
         \xi_W=0 & u_B^0 & u_B^1  \\
        \hline
          u_A^0& -1 & 0\\
          u_A^1 & 0 & -1 \\
            \hline
    \end{array}\
    \begin{array}{|c|c|c|}
         \hline
       \xi_W=1 &  u_B^0  & u_B^1  \\
        \hline
          u_A^0& 0 & -3/4 \\
          u_A^1& -3/4 & 0 \\
            \hline
    \end{array}
    \label{eq:egCACinstance}
\end{equation}
When $ \xi_W=0,$ it is beneficial for the players to concentrate the mass of $ Q(\cdot|\cdot) $ on the subset  $\{(u_A^0,u_B^0), (u_A^1,u_B^1)\}$, whereas when $ \xi_W=1 $, it is beneficial to do so on the complement $ \{(u_A^1,u_B^0),(u_A^0,u_B^1)\}.$ We refer to this situation as the \textit{coordination dilemma}. The dilemma is about whether they should match the index of their actions, \ie `coordinate', or mismatch these indices, or  `anti-coordinate'.

In a hypothetical centralized setting where both players had access to both $ \xi_A $ and $ \xi_B $, players could agree on the `best' estimate of $ \xi_W $ and choose to either coordinate or anti-coordinate. But the decentralized nature of the problem implies that players could have differing views about the value of $ \xi_W $, thereby leading to the above dilemma.

It is intuitive that if players could correlate their actions such a way that reflects an optimal midway compromise between coordination and anti-coordination, they could potentially achieve a better cost on average than either coordination or anti-coordination.
Unfortunately, with classical strategies, possibilities for creating such correlation is limited. In fact, even with the optimal classical strategies, players cannot do better than they could \textit{without} randomization. However, remarkably, we showed in \cite{deshpande2022quantum} that players can do better through \textit{quantum randomization} obtained by correlating their actions through \textit{entanglement}. For the above problem, we demonstrated a physically realizable quantum strategy that strictly outperforms all classical strategies, thereby showing the existence of a quantum advantage in decentralized control.

\subsection{Problem classes}
A decision problem in our setting is  specified by the prior distribution $\Pbb(\xi_A, \xi_B, \xi_W)$ on the states of nature, the action spaces $\mathcal{U}_A, \mathcal{U}_B$ and the cost function $\ell$.
We assume that  $\ell$ satisfies
\begin{equation}\label{eq:supstruct}
\ell(u_A, u_B, \xi_W=0)\in\{0, -1\},\  \ell(u_A, u_B, \xi_W=1)\in\{0, -\chi\},
\end{equation}
where $ \chi > 0 $ is a parameter.
Such `binary' costs capture settings where a fixed cost is incurred based on whether an underlying event (such as successful transmission of a packet) occurs or does not, given the background state $ \xi_W $.
The parameter $\chi$ captures the degree  to which the costs differ depending on $\xi_W$. With the specification in \eqref{eq:supstruct}, we now construct a superstructure of subclasses.
\begin{definition} \label{def:superstruc}
Let $M,N \in \{0,-1\}^{2\times 2}$ be matrices with each entry in $\{0, -1\}$.
A problem class $\Cscr(M,N)$ specified by the tuple  $(M, N)$ is the set
\begin{align}
\Cscr(M,N):=&\{D|D=(M, N, \Pbb, \Uscr_A,\Uscr_B, \chi)\ {\it where}\nonumber \\
&\ \Pbb\in \Pscr(\Xi), |\Uscr_A|=|\Uscr_B| =2\ {\it and}\ \chi\in [0,\infty)\}
\label{eq:classdef}
\end{align}
\def\Crsfs{\mathscr{C}}
Denote the set of all problem classes in our superstructure by $\Crsfs$.

An element $D=(M, N, \Pbb, \Uscr_A,\Uscr_B, \chi)\in \Cscr(M,N)$ is called a problem instance. The cost function $ \ell:\Uscr_A\times\Uscr_B\times\Xi_W\to \{0, -1, -\chi\} $ of this instance is given by
\begin{equation}
	\ell(u_A^i, u_B^j, 0)=[M]_{i+1\ j+1},\quad \ell(u_A^i, u_B^j, 1)=\chi[N]_{i+1\ j+1} \label{eq:costmat}
\end{equation}
\end{definition}

\begin{definition}[CAC class] \label{def:cac}
	The problem class $\Cscr(M,N)$ with $M,N$ as
\begin{equation}\label{eq:cacdef}
	M=\begin{pmatrix}
		-1&0\\
		0&-1
	\end{pmatrix} \text{and} N=\begin{pmatrix}
		0&-1\\
		-1&0
	\end{pmatrix},
\end{equation}
is referred to as the coordinate-anti-coordinate class, or in short, CAC class. $(M,N)$ satisfying \eqref{eq:cacdef} are said to be \textit{CAC form}.
\end{definition}
Notice that the cost described in equation \eqref{eq:egCACinstance} belongs to the CAC class with
and $ \chi = \frac{3}{4} $.

\begin{definition}[$\half$-CAC class] \label{def:halfcac}
The problem class $\Cscr(M,N)$ with $M,N$ as
	\begin{equation}\label{eq:halfcacdef}
		M=\begin{pmatrix}
			-1&0\\
			0& 0
		\end{pmatrix} \text{and} N=\begin{pmatrix}
			0&-1\\
			-1&0
		\end{pmatrix},
	\end{equation}
	is referred to as the half coordinate-anti-coordinate class, or in short, $\half$-CAC class. $(M,N)$ satisfying \eqref{eq:halfcacdef} are said to be \textit{$\half$-CAC form}.
\end{definition}

\begin{comment}
\begin{table}[H]
    \centering
    \begin{tabular}{|c|c|c|}
    \hline
         $\xi_W=0$& $U$ & $D$  \\
        \hline
          $L$& $\ell(L, U, 0)$ & $\ell(L, D, 0)$ \\
          $R$& $\ell(R, U, 0)$ & $\ell(R, D, 0)$ \\
            \hline
    \end{tabular}
     \subcaption{$\ell(u_A, u_B, \xi_W=0)\equiv M$}
    \begin{tabular}{|c|c|c|}
         \hline
         $\xi_W=1$& $U$ & $D$  \\
        \hline
          $L$& $\ell(L, U, 1)$ & $\ell(L, D, 1)$ \\
          $R$& $\ell(R, U, 1)$ & $\ell(R, D, 1)$ \\
         \hline
    \end{tabular}
    \subcaption{$\ell(u_A, u_B, \xi_W=1)\equiv \chi N$}
    \caption{Cost function of a problem instance}
    \label{eq:costagentrel}
\end{table}
\end{comment}
Each of the four entries in matrices $M$ and $N$ are allowed to take binary values in $\{0, -1\}$ for  instances in our superstructure. We thus have $2^4\times 2^4=256$ possible problem classes in our superstructure.

\subsection{Decentralised estimation: a motivating example}
We motivate our investigation through a concrete example of a decentralised estimation problem we introduced in~\cite{deshpande2022quantum}. Suppose that the agents $A$ and $B$ collaborate to produce an estimate $f:\Uscr_A\times\Uscr_B\to\Xi_W$ of $\xi_W$, given their local information. The choice of such an $f$ determines how their actions collate to produce the desired estimate, and thereby shapes the cost function. With $\chi(0):=1$ and $\chi(1):=\chi$, suppose that the cost is given by
\begin{equation}
\ell(u_A, u_B, \xi_W)=-\chi(\xi_W)\delta(\xi_W, f(u_A, u_B)).
\end{equation}
We now find that the choice of $f$, dictated by some estimation `mechanism', allots this estimation problem to a problem class within our superstructure. For instance if $f(u_A^i, u_B^j)=i\oplus j$, we find that the problem is an instance of the CAC class with $M$ and $N$ as given by \eqref{eq:cacdef}.
On the other hand, if $f(u_A^i, u_B^j)=i\vee j$, then the problem belongs to another class $\Cscr(M, N)$ with
\begin{equation}\label{eq:orestimate}
M=\begin{pmatrix}
			-1&0\\
			0& 0
		\end{pmatrix} \text{and} N=\begin{pmatrix}
			0&-1\\
			-1&-1
		\end{pmatrix}.
\end{equation}
In our previous article \cite{deshpande2022quantum}, we found that instances of the estimation problem with $f(u_A^i, u_B^j)=i\oplus j$ admit a quantum advantage. In particular, we find that quantum strategies enable the two agents to effectively collaborate during game-play, when such collaborations are restricted within the classical realm of passive common randomness. It is therefore of interest to examine what aggregations $f$ induce a cost structure that admits such a quantum advantage. Our investigation through this two-part series answers this query in a reasonably detailed manner. In this particular context of estimation, our analysis reveals that the cost structure \eqref{eq:orestimate} induced by the aggregate $f(u_A^i, u_B^j)=i\vee j$ does not admit a quantum advantage while that \eqref{eq:cacdef} induced by $f(u_A^i, u_B^j)$ does.

\section{Decision Strategies and Non-Local Advantages}\label{sec:nonlocal}
\subsection{Decision strategies}
We study decision problems in the above superstructure in space of stochastic policies that specify a probability distribution on $ \Uscr $ given the information of both players; in the classical Markov decision processes setting, these reduce to what are known as occupation measures \cite{borkar88convex} and have been employed in other information structures as well~\cite{kulkarni2014optimizer}. We refer the reader to our earlier work \cite{deshpande2022quantum} for more details. Under this framework, any decision strategy is described as a joint conditional probability distribution $Q\in \Pscr(\Uscr|\Xi)$
%\begin{equation}
%Q: u_A, u_B, \xi_A, \xi_B\in \mathcal{U}_A\times \mathcal{U}_B\times \Xi_A\times \Xi_B \to Q(u_A, u_B|\xi_A, \xi_B)\in [0, 1],
%\end{equation}
that is required to satisfy a certain set of constraints. Based on these constraints we have classes $ \Nscr\Sscr , \Gscr$  and $ \Qscr $ defined below.
In each case the expected cost of a problem instance $D$ under a strategy $Q$ is given by
\begin{align}
J(Q;D)&=\sum_{\xi_A, \xi_B, \xi_W}\sum_{u_A,u_B}\Pbb(\xi_A, \xi_B, \xi_W) \non\\
&\qquad \times \ell(u_A, u_B, \xi_W)Q(u_A, u_B|\xi_A, \xi_B).\label{eq:expcost}
\end{align}
Any distribution $Q$ by virtue of belonging to $\Pscr(\Uscr|\Xi)$, regardless of the strategic class under consideration, satisfies the  positivity and normalisation constraints for probability distributions,
\begin{equation}\label{eq:normalise}
Q(u|\xi)\geq 0\ \forall u, \xi,\quad \sum_u Q(u|\xi)=1\ \forall\ \xi.
\end{equation}
We investigate across three different strategic classes, each described by further restrictions on $ Q $.

\subsubsection{Local distributions} Set of \textit{local distributions} ($\Dscr$) is the set of distributions $Q\in \Pscr(\Uscr|\Xi)$ that correspond to locally randomized strategies. In a locally randomized strategy (also called \textit{behavioural strategy} in game theory), decision maker $i$ assigns a conditional probability distribution $Q_i\in\Pscr(\Uscr_i|\Xi_i)$ on his actions given his information. Then $Q$ assumes the form
\begin{equation}
Q(u_A, u_B|\xi_A, \xi_B)=Q_A(u_A|\xi_A)Q_B(u_B|\xi_B)\label{eq:locallyrandom}
\end{equation}
An important subset of local distributions is the set of deterministic strategies, $\Pi$, which is the set of all strategies $Q\in\Dscr$ expressible as
$$Q(u|\xi)=\delta(u_A, \gamma_A(\xi_A))\delta(u_B, \gamma_B(\xi_B))$$
where $\gamma_i:\Xi_i\to\Uscr_i$ for each $i=A, B$.
\subsubsection{Local polytope ($ \Lscr $)}
The local polytope $\Lscr$ is the set of all classical strategies implementable through an arbitrary amount of passive common randomness. It is the set of all $Q\in\Pscr(\Uscr|\Xi)$ expressible as
$$Q=\sum_{\omega\in\Omega} \Phi(\omega)\prod_i Q(u_i|\xi_i)$$
for a finite set $\Omega$, and distributions $\Phi\in\Pscr(\Omega)$ andd $Q_i\in\Pscr(\Uscr_i|\Xi_i)$ for each $i\in A, B$. By definition, one can note that $\Lscr=\conv(\Dscr)$. In fact, $\Lscr=\conv(\Pi)$.
We refer the reader to \cite{deshpande2022quantum, saldi2022geometry} for more details.

%Let $\Gamma_i$ be the set of functions $\gamma_i: \Xi_i\to \Uscr_i$. We call a tuple $\gamma:=(\gamma_A, \gamma_B)$ a deterministic strategy, and $\Gamma$ is the set of all such strategies. For every problem class in our superstructure $|\Gamma|=16$. The policy $\pi_\gamma\in\Pscr(\Uscr|\Xi)$ corresponding to a deterministic strategy $\gamma$ is given by
%\begin{equation}
%\pi_\gamma(u_A, u_B|\xi_A, \xi_B)=\delta_{u_A,\gamma(\xi_A)}\delta_{u_B,\gamma(\xi_B)}.\label{eq:detoccupationmeasure}
%\end{equation}
%The class of \textit{mixed strategies} is where the decision makers jointly assign a probability distribution $P\in \Pscr(\Gamma)$ on the set of deterministic strategies. The occupation measure corresponding to a mixed strategy specified by $P\in\Pscr(\Gamma)$ is
%\begin{equation}
%Q_P(u_A, u_B|\xi_A, \xi_B)=\sum_{\gamma_A, \gamma_B}P(\gamma_A, \gamma_B)\pi_\gamma(u_A, u_B|\xi_A, \xi_B)
%\end{equation}
%The set of all such distributions,
%\begin{equation}
%\Lscr:=\{Q_P\in \Pscr(\Uscr|\Xi)| P\in \Pscr(\Gamma)\},
%\end{equation}
%is the convex hull of the set
%\begin{align}
%\Pi:=\{\pi_\gamma\in\Pscr(\Uscr|\Xi)| &\pi_\gamma(u_A, u_B|\xi_A, \xi_B)\nonumber\\
%&=\delta_{u_A,\gamma(\xi_A)}\delta_{u_B,\gamma(\xi_B)},
% \gamma\in\Gamma\}.\label{eq:localsimplex}
%\end{align}
%We refer to $\Lscr$ as \emph{the local polytope}.

\subsubsection{Quantum ellitope $\Qscr$}The \textit{quantum ellitope} denoted $\mathcal{Q}$ is the set of distributions  $Q\in \Pscr(\Uscr|\Xi)$ generated by quantum strategies. We mathematically specify a quantum strategy as a tuple $Q=(\Hscr_A, \Hscr_B, \rho_{AB}, \{P_{u_A}^A(\xi_A)\}, \{P_{u_B}^B(\xi_B)\})$ where

\textbf{(i)} $\Hscr_A$ and $\Hscr_B$ are finite dimensional Hilbert spaces.

 \textbf{(ii)} $\rho_{AB}\in \Bscr(\Hscr_A\otimes\Hscr_B)$ is a density matrix, \ie $\rho_{AB}\succeq 0$ and $\Tr(\rho_{AB})=1$.
%\textbf{(ii)} We choose $A(\xi_A)$ and $B(\xi_B)$ as sets of basis vectors for $ \Hscr_A $  and  $\Hscr_B$ as constituting a measurement basis for each value of $\xi_A,\xi_B$ for decision makers $A,B$. Since each $ \xi_i$ is binary, each decision maker $i$ has two sets of orthonormal basis vectors, each set containing $ \dim(\Hscr_i) $ vectors.
%    Denote these bases as follows
%    \begin{equation}
%     A(\xi_A):=\{\ket{v_i(\xi_A)_A}\}_i\quad \forall \xi_A,\label{eq:basis1}
%     \end{equation}
%     \begin{equation}
%     B(\xi_B):=\{\ket{w_j(\xi_B)_B}\}_j\quad \forall \xi_B.\label{eq:basis2}
%     \end{equation}

\textbf{(iii)} For each $i\in\{A, B\}$,  $P_{u_i}^i(\xi_i)\in\Bscr(\Hscr_i)$ are projection operators that obey $P_{u_i}^i(\xi_i)^2=P_{u_i}^i(\xi_i)$ and $\sum_{u_i} P_{u_i}^i(\xi_i)=\Ibf_i$ for each $\xi_i\in\Xi_i$ where $\Ibf_i\in\Hscr_i$ is the identity operator.

%Define $ \{A^{u_A}(\xi_A)\}_{u_A \in \Uscr_A} $ as a partition of $ A(\xi_A) $ which satisfies
%   \begin{align}
%   A^{u_A}(\xi_A)\cap A^{u'_A}(\xi_A)=&\emptyset \quad \forall u_A\neq u'_A;\nonumber\\
%    \bigcup_{u_A\in\mathcal{U}_A}A^{u_A}(\xi_A)&=A(\xi_A).\label{eq:mece}
%   \end{align}
%Define projection operators onto $ A^{u_A}(\xi_A)$ for each action $u_A \in \Uscr_A$ as follows,
%    \begin{equation}
%    P_{u_A}^{(A)}(\xi_A):=\sum_{\ket{v{(\xi_A)_A}}\in A^{u_A}(\xi_A)} \ket{v{(\xi_A)_A}}\bra{v{(\xi_A)_A}}.\label{eq:projectorsB}
%    \end{equation}
%Similarly define projection operators for the other decision maker $B$,
%   \begin{equation}
%    P_{u_B}^{(B)}(\xi_B):=\sum_{\ket{v{(\xi_B)_B}}\in B^{u_B}(\xi_B)} \ket{w{(\xi_B)_B}}\bra{w{(\xi_B)_B}}.\label{eq:projectorsH}
%    \end{equation}
% Note that if $ A^{u_A}(\xi_A) $ is empty, then  $P_A^{u_A}(\xi_A)$ is equal to the null matrix on $\mathcal{H}_A$. \\
%Points (i)-(ii) above describe an \textit{experimental setup} and (iii) describes the \textit{measurements} that can be performed on it and the actions that have to be taken as a function of the outcome these measurements.
The described quantum strateguy $Q$ renders the following occupation measures.
\begin{equation}
Q(u_A, u_B|\xi_A, \xi_B)=\Tr\left(P^{(A)}_{u_A}(\xi_A)\otimes P^{(B)}_{u_B}(\xi_B)\rho_{AB}\right). \label{eq:quantconditional}
\end{equation}
Thus $ \Qscr $ is the set of all distributions $ Q $ that satisfy \eqref{eq:quantconditional} for some choice of the tuple $(\Hscr_A, \Hscr_B, \rho_{AB}$, $\{P^{(A)}_{u_A}(\xi_A)\}_{u_A, \xi_A}$, $\{P^{(B)}_{u_B}\}_{u_B, \xi_B})$. A detailed discussion on the physical implementation of quantum strategies during gameplay can be found in \cite{deshpande2022quantum}. Nevertheless, the expected cost of a problem $D$ under the strategy $Q$ is given by (from \eqref{eq:quantconditional} and \eqref{eq:expcost})
\begin{align}
	J(Q; D)
	&=\sum_{\xi_A, \xi_B, \xi_W}\Pbb(\xi_A, \xi_B, \xi_W)\sum_{ u_A, u_B} \ell(u_A, u_B, \xi_W) \non \\
	&\times\Tr\left(\rho_{AB}P_{u_A}^{(A)}(\xi_A)P_{u_B}^{(B)}(\xi_B)\right) \label{eq:qcost}
\end{align}
\subsubsection{No-signalling polytope} The set of \textit{no-signalling distributions}, denoted $\mathcal{NS}$, is the set of distributions $Q\in \Pscr(\Uscr|\Xi)$ that satisfy the following \textit{no-signalling} constraints that prohibits communication between the two agents~\cite{matthews2012nonsigcode, saldi2022geometry}, as demanded by the stasis of the information structure. \\
We have $\forall {u_A, \xi_A, \xi_B, \xi_B'}$
\begin{equation}
    \sum_{u_B} Q(u_A, u_B|\xi_A, \xi_B)
    =\sum_{u_B'} Q(u_A, u_B'|\xi_A, \xi_B');\label{eq:nosig1-222}
\end{equation}
and $\forall {u_B, \xi_A, \xi_A', \xi_B}$
\begin{equation}
    \sum_{u_A} Q(u_A, u_B|\xi_A, \xi_B)
    =  \sum_{u_A'} Q(u_A', u_B|\xi_A', \xi_B); \label{eq:nosig2-222}
\end{equation}
This asserts that the choice of conditional distribution of one agent given his information does not affect the outcome distribution of the other agent, and thus the joint distribution respects the prohibition of \textit{faster than light communication}. Since this set of distributions is characterised by a finite number of linear equalities and inequalities, $ \Nscr\Sscr $ a polytope.

\subsubsection{Centralised polytope} We call the whole set of conditional distributions on $\Pscr(\Uscr|\Xi)$ the \textit{centralised polytope}.
\subsection{Advantages}
The following proposition shows that the quantum ellitope includes all local distributions.
\begin{proposition}\label{prop:qsupersetpi}
For every $\pi_{\gamma}\in\Pi$ specified by some deterministic strategy $\gamma$, there exists a $ Q\in\Qscr$ such that $\pi_\gamma\equiv Q$. Thus,
\begin{equation}\label{eq:heirarchy}
\mathcal{L}\subset\mathcal{Q}\subset\mathcal{NS}\subset \Pscr(\Uscr|\Xi).
\end{equation}
\end{proposition}
\begin{proof}
Choose $\Hscr_A$, $\Hscr_B$ and $\rho_{AB}\in\Bscr(\Hscr_A\times\Hscr_B)$. Take $P_{u_i}^{(i)}(\xi_i)=\delta_{u_i\gamma_i(\xi_i)}\Ibf$ so that
\begin{align}
Q(u_A, u_B|\xi_A, \xi_B)&=\Tr(P^{(A)}_{u_A}(\xi_A)\otimes P^{(B)}_{u_B}(\xi_B)\rho_{AB})\nonumber\\
&=\delta_{u_A\gamma_A(\xi_A)}\delta_{u_B\gamma_B(\xi_B)}\nonumber\\
&=\pi_\gamma(u_A, u_B|\xi_A, \xi_B),
\end{align}
for all $u_A,u_B,\xi_A,\xi_B$. The convexity of $\Qscr$~\cite{deshpande2022quantum} completes the proof.
\end{proof}

For $S=\Lscr, \Pi, \Qscr, \Nscr\Sscr$ and $\Pscr(\Uscr|\Xi)$, define
\begin{equation}
J^*_{S}(D):=\inf_{Q\in S} J(Q;D)
\end{equation}
and denote the centralised optimum $J^*_{\Pscr(\Uscr|\Xi)}$ by $J^{**}(D)$
Note that since $\Lscr=\conv(\Pi)$ and $J(Q;D)$ is linear in $Q$, $J^*_\Pi(D)=J^*_\Lscr(D)$. Then from \eqref{eq:heirarchy} the following relationship holds between the respective infima,
\begin{equation}\label{eq:costheirarchy}
J^*_\Pi(D)=J^*_{\mathcal{L}}(D)\geq J^*_{\mathcal{Q}}(D)\geq J^*_{\mathcal{NS}}(D)\geq J^{**}(D).
\end{equation}
\begin{definition}
We say that the problem class $\Cscr(M,N)$
\begin{enumerate}
	\item admits a \textit{quantum advantage} if $\exists D\in \Cscr(M,N): J^*_{\mathcal{L}}(D)> J^*_{\mathcal{Q}}(D)$.
	\item admits a \textit{no-signalling advantage} if  $\exists D\in\Cscr(M, N): J^*_{\mathcal{L}}(D)> J^*_{\mathcal{NS}}(D)$.
	\item admits a \textit{centralisation advantage} if  $\exists D\in\Cscr(M, N): J^*_{\mathcal{L}}(D)> J^{**}(D)$.
\end{enumerate}
\end{definition}
For any instance $D$, we have from \eqref{eq:costheirarchy} that if $J^*_{\mathcal{L}}(D)= J^*_{\mathcal{NS}}(D)$ then  $J^*_{\mathcal{L}}(D)= J^*_{\mathcal{Q}}(D)= J^*_{\mathcal{NS}}(D)$.  Hence if a problem class $\Cscr(M, N)$ does not admit a no-signalling advantage then it does not admit a quantum advantage.  Similarly $J^*_{\mathcal{L}}(D)= J^{**}(D)$ forces $J^*_{\mathcal{L}}(D)= J^*_{\mathcal{Q}}(D)= J^*_{\mathcal{NS}}(D)=J^{**}(D)$ and thus if a problem class $\Cscr(M, N)$ does not admit a centralisation advantage then it does not admit a no-signalling and a quantum advantage.

\subsection{No-signalling polytope and Bell inequalities}\label{sec:bell}
At some level our work is directly in correspondence with the state of the art in the geometry of quantum correlations.
The description of the quantum ellitope remains largely abstract to this date, though it is known to be convex and non-polytopic~\cite{ramanathan2016noquantum}, \cite{werner2001multipartitebell}. A convergent hierarchy of semi-definite programs is known that characterizes the set~\cite{navascues2008hierarchy}.
More is known about the no-signalling polytope $\Nscr\Sscr$  that contains the quantum ellitope~\cite{popescu1994nonlocquantum}. In our case, since $|\Uscr_i|=\Xi_i|=2$ for each $i$, the   $\Nscr\Sscr$ has 24 vertices, 16 of which are local and correspond to the vertices of $\Lscr$ which constitute the set of deterministic strategies, $\Pi$. The following proposition concisely enumerates these vertices. We refer the reader to~\cite{barrett2005nonlocal} for further discussion.
\begin{proposition}
\begin{enumerate}
	\item The set $\Pi$ of deterministic strategies is given by
	\begin{equation}
		\pi^{\alpha\gamma\beta\delta}(u_A^i, u_B^j|\xi_A, \xi_B)=\begin{cases}
			1 & i=\alpha \xi_A \oplus \beta\\
			& j= \gamma\xi_B\oplus \delta\\
			0 & \text{otherwise.}
		\end{cases}.\label{eq:localvertex}
	\end{equation}
where $\alpha, \beta, \gamma, \delta\in \{0, 1\}$.
	\item The 8 non-local vertices of  $\Nscr\Sscr$ are given by
	\begin{multline}
		Q^{\alpha\beta\delta}(u_A^i, u_B^j|\xi_A, \xi_B)\\
		=\begin{cases}
			1/2 & i\oplus j=\xi_A.\xi_B\oplus\alpha \xi_A\oplus\beta \xi_B\oplus\delta\\
			0 & \text{otherwise.}
		\end{cases}\label{eq:nonlocalvertex}
	\end{multline}
where $\alpha, \beta, \delta\in \{0, 1\}$.
\end{enumerate}
\end{proposition}

%For control theorists primarily interested in harnessing the power of quantum correlations for decentralised control, this abstraction stands as quite a hurdle, not allowing for a general isolation of linear objectives that do admit a quantum advantage.

In quantum information theory, the non-local nature of quantum correlations is principally captured by their violation of a sum of experimentally testable correlations, known as the \textit{Bell inequalities}~\cite{bell1964epr}. Geometrically, non-locality implies that the first inclusion in \eqref{eq:heirarchy} is strict; thus Bell inequalities linearly separate $\Lscr$ from some point in $\Qscr$.
%Further notice that we have $2^{|\Xi_A|+|\Xi_B|+|\Uscr_A|+|\Uscr_B|}=16$ free parameters $Q(u_B, u_H|\xi_B, \xi_H)$ in any policy. Normalisation \ref{eq:normalise} and the no-signalling conditions~\eqref{eq:nosig1-222}, \eqref{eq:nosig2-222} impose 8 constraints and we thus have the $\Nscr\Sscr$ and $\Lscr$ polytopes, as well as the quantum ellitope $\Qscr$ residing in an 8 dimensional space.
One of the faces of the local polytope which is not a face of the no-signalling polytope is described by the popular CHSH inequality \cite{clauser1969chsh},  \cite{goh2018geometry}, which is a generalization of Bell's original inequality. We illustrate this geometry of the CHSH inequality in figure \ref{fig:polytopes}. We direct the reader to \cite{goh2018geometry} for a deeper peek into the relative geometry of $\Lscr$, $\Qscr$ and $\Nscr\Sscr$.

\begin{figure}
\centering
\includegraphics[scale=0.4]{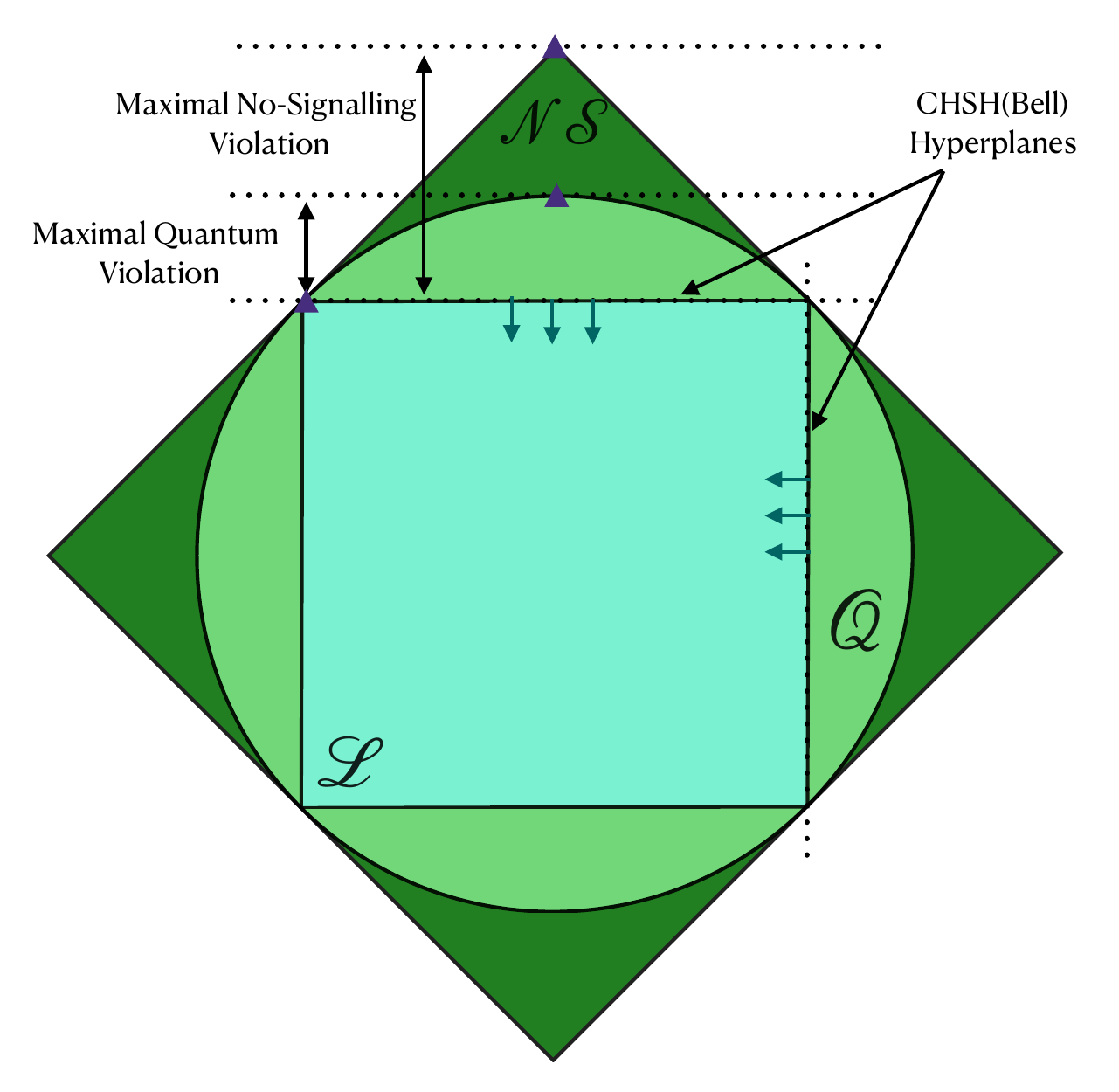}
\caption{Geometry of $\Nscr\Sscr$, $\Lscr$ and $\Qscr$ and violation of a CHSH (Bell) inequality.}\label{fig:polytopes}
\end{figure}

%In our previous work~\cite{deshpande2022quantum} we showed how a general static decision problem is an optimisation problem of a linear objective $\sum_{u, \xi} \Jscr(u, \xi)Q(u|\xi)$ for $Q$ within some set of conditional distributions specified by the strategic class under consideration. With $Q$ constrained to lie within the polytope $\Lscr$ for classical strategies and the ellitope $\Qscr$ for quantum ones, these objectives admit corresponding infimas $J_\Lscr^*$ and $J_\Qscr^*$ respectively. We are interested in the investigation of the quantum advantage  $J_\Lscr^*-J_\Qscr^*$ for a given decision problem. Given the abstract specification of the set of quantum correlations, it is in general difficult to analytically investigate a general objective. We thus investigate over different classes $\Cscr(M,N)$ in our superstructure and isolate problem classes, and some regimes within them where the quantum advantage is non-zero. It is clear from \eqref{eq:costheirarchy} that an absence of a \textit{no-signalling} advantage in a given decision problem implies the absence of any quantum advantage, and thus optimality of a deterministic decision strategy. Owing to the polytopic geometry of $\mathcal{NS}$ and $\mathcal{L}$, it is also simpler to investigate no-signalling advantage. We first look at equivalences in our superstructure that help truncate the set of problem classes that require such a line of distinct investigation.

\section{Equivalences in the Class Superstructure}\label{sec:equivalences}
Recall that our superstructure comprises of $256$ different problem classes, each corresponding to a binary matrix tuple $(M,N)$.
%We wish to eliminate all classes $\Cscr(M, N)$ that do not admit the quantum advantage. This requires systematic scanning and elimination though all $256$ different tuples $(M, N)$.
In this section we establish a set of equivalences between the existence of a quantum advantage across problem classes. The proofs of these propositions are relegated to Appendix \ref{app:equiproofs}.

To begin, Proposition~\ref{prop:transpose} asserts such an equivalence between a problem class $\Cscr(M\t, N\t)$ and the $\Cscr(M, N)$, for an arbitrary tuple $(M, N)$. This is intuitive since transposition of $M,N$ in fact corresponds to exchanging the two agents in the problem instance.
\begin{proposition}  \label{prop:transpose}
	\textit{(Transposition equivalence (exchange of agents))}
 The following are equivalent:
\begin{enumerate}
	\item $\Cscr(M,N)$ admits a quantum advantage.
	\item  $\Cscr(M\t, N\t)$ admits a quantum advantage.
\end{enumerate}
\end{proposition}

Recall the matrix $\Rsf$ from \eqref{eq:rdef}. We show in Proposition~\ref{prop:permute} that the existence of a quantum advantage in $\Cscr(\Rsf M, \Rsf N)$ and likewise in $\Cscr(M\Rsf, N\Rsf)$ is equivalent to that in $\Cscr(M, N)$. The tuples $(\Rsf M, \Rsf N)$ and  $( M\Rsf, N\Rsf)$ correspond to an  exchange of rows and columns in $(M, N)$, respectively. An exchange of rows is tantamount to relabelling the actions of player $A$ as $(u_A^0, u_A^1)\mapsto (u_A^1, u_A^0)$, and those of columns corresponds to a similar relabelling for player $B$.
\begin{proposition}
\label{prop:permute}
\textit{(Permutation equivalence (relabelling of actions))}
Let $\Rsf$ be as defined \eqref{eq:rdef}. Then the following are equivalent:
\begin{enumerate}
	\item  $\Cscr(\mathsf{R}M, \mathsf{R}N)$ admits quantum advantage.
	\item  $\Cscr(M\mathsf{R}, N\mathsf{R})$ admits quantum advantage.
	\item  $\Cscr(M,N)$ admits quantum advantage.
\end{enumerate}\end{proposition}

%2) ($2\longleftrightarrow 3$) For $\Cscr(M, N)\ni D= (M, N, \Pbb, \Uscr_A, \Uscr_B, \chi)$, consider a $\Cscr(M\mathsf{R}, N\mathsf{R})\ni D'=(M\mathsf{R}, N\mathsf{R}, \Pbb, \Uscr_A, \Uscr_B', \chi)$ with $\Uscr_B':=\{u_B^{0\prime}, u_B^{1\prime}\}=\{u_B^1, u_B^0\}$. Then for any strategy $Q$ irrespective of the strategic class,
%\begin{align*}
%&J(Q; D')\\
%&=\sum_{\xi, u_A}\Pbb(\xi_A, \xi_B, \xi_W)\sum_{i}\ell'(u_A, u_B^{i\prime}, \xi_W)Q(u_A, u_B^{i\prime}|\xi_A, \xi_B)\\
%&=\sum_{\xi, u_A}\Pbb(\xi_A, \xi_B, \xi_W)\sum_{i}\ell(u_A, u_B^i, \xi_W)Q(u_A, u_B^i|\xi_A, \xi_B)\\
%&=J(Q; D)
%\end{align*}
%Thus,
%$$J(Q, D)-J_\Pi^*(D)<0\longleftrightarrow J(Q, D')-J_\Pi^*(D')<0$$
%and we have shown $2\longleftrightarrow 3$.

Finally we have Proposition \ref{prop:exchange} asserting the equivalence of $\Cscr(N, M)$ and $\Cscr(M, N)$. This corresponds to relabelling the values of $\xi_W$.
%Indeed, since a problem class subsumes instances with all distributions $\Pscr(\Xi)$, it is possible to define a new instance in the problem class $\Cscr(N, M)$ by relabelling $\xi_W$ in an instance $D\in\Cscr(M, N)$. The proof idea is exactly this.
\begin{proposition}
\label{prop:exchange}
\textit{(Exchange equivalence (relabelling of $\xi_W$))}
The following are equivalent:
\begin{enumerate}
	\item  $\Cscr(M, N)$ admits quantum advantage.
	\item  $\Cscr(N,M)$ admits quantum advantage.
\end{enumerate}\end{proposition}
\section{Main result}\label{sec:eliminate}
Let $V:=(M, N)$ and define actions $\Tsf,\Rsf,\Rsf',\Esf$ so that $\Tsf V:=(M\t, N\t)$, $\Rsf V=(\Rsf M, \Rsf N)$, $\Rsf'V:= (M\Rsf, N\Rsf)$ and $\Esf V:= (N, M)$. Let \begin{equation}\label{eq:groupactions}
	\Omega:=\{\Ibf, \Tsf, \Rsf, \Rsf', \Esf\},
\end{equation}
be the set of group actions with $\Ibf$ being the identity and denote by $(V; \Omega)$ the matrix-pairs generated by an arbitrary sequence of group actions on $V$ (technically, the \textit{orbit} of $V$ under $\Omega$). 
Then using Propositions \ref{prop:transpose}-\ref{prop:exchange}, then if $\Cscr(V)$ does not admit a quantum advantage, then $\Cscr(X)$ does not admit quantum advantage for all $X\in (V; \Omega)$. We henceforth use the notation
$$\Cscr((V; \Omega)):=\{\Cscr(X): X\in (V;\Omega)\},$$
and refer to this as the \textit{orbit} of the class $V$.
Following is the main theorem of this paper.
\begin{theorem}\label{thm:main}
	Consider the problem class superstructure defined in Definition~\ref{def:superstruc}. A problem class $\Cscr(A,B)$ in this superstructure admits a quantum advantage if and only if $(A,B) \in ((M,N);\Omega)$ where $(M,N)$ are either in the CAC form or in the  $\half$-CAC form.
\end{theorem}

Thus a problem class admits a quantum advantage if and only if it lies in the orbit of the CAC class or the $\half$-CAC class. We now proceed to prove this claim. In our earlier paper~\cite{deshpande2022quantum}, we showed that the CAC class admits a quantum advantage. In Section~\ref{sec:1/2cacpoc} we show that the $\half$-CAC class admits a quantum advantage. In the sections below we systematically eliminate all classes not in the orbit of the CAC and $\half$-CAC class to show Theorem~\ref{thm:main}. 
\def\Crsfs{\mathscr{C}}

\begin{definition}
	We call a problem class $\Cscr(M,N)\in\Crsfs$ an $m$-$n$ class if the number of non-zero entries in $M$ is $m\in\{0,1,2,3,4\}$ and the number of non-zero entries in $N$ is $n\in\{0,1,2,3,4\}$ and call $V=(M, N)$ an $m$-$n$ tuple. Let $\Cscr_{mn} \subset \Crsfs$ denote the set of all $m$-$n$ problem classes. 
\end{definition}
Notice that $|\Cscr_{mn}|=\prescript{4}{}C_m \prescript{4}{}C_n$ and $\{\Cscr_{mn}\}_{m,n}$ defines a partition on $\Crsfs$. For $(M, N)$ in the CAC form, $\Cscr((M, N);\Omega)\subset \Cscr_{22}$. Similarly, for $(M, N)$ in the $\half$-CAC form, $\Cscr((M, N);\Omega)\in \Cscr_{12}\cup \Cscr_{21}$. To proceed with our systematic elimination, we eliminate $\Cscr_{mn}$ for all pairs $(m, n)\notin \{(2,2), (1,2), (2,1)\}$ through a pigeonhole principle-based argument; this is done in Section \ref{sec:cent}. In the subsequent sections, we eliminate classes in $\Cscr_{22}$, $\Cscr_{12}$ and $\Cscr_{21}$ that do not belong in the orbit of the CAC or the $\half$-CAC class.

\subsection{Problem classes with no centralisation advantage} \label{sec:cent}
Since $\ell$ takes only binary values, if there exists a pair of actions $u_A^*, u_B^*$ such that $\ell(u_A^*, u_B^*,\xi_W)$ is non zero for both values of $\xi_W$, then the strategy 
\begin{equation}\label{eq:ccq}
	\bar{Q}(u_A,u_B|\xi_A,\xi_B) \equiv  \delta_{u_A=u_A^*,u_B=u_B^*}
\end{equation}
 which lies in $\Pi$ is optimal over $\Pscr(\Uscr|\Xi)$. In other words, the problem admits no centralisation advantage. The following definition and the proposition that succeeds formalises this line of arguments.

\begin{definition}
We call a pair $V=(M, N)$ overlapping if $\exists i, j $ such that $ [M]_{ij}=[N]_{ij}=-1$. Denote the set of all classes $\Cscr(V)$ where $V$ is overlapping by $\Cscr^o$,
\begin{equation}\label{eq:codef}
	\Cscr^o:=\{\Cscr(M, N)|\exists i, j: [M]_{ij}=[N]_{ij}=-1\}.
\end{equation}
\end{definition}
\begin{lemma}\label{prop:centralproblems}
If $\Cscr(M,N)\in\Cscr^o$, then $\Cscr(M,N)$ does not admit a centralisation, and hence a quantum advantage. 
\end{lemma}
\begin{proof}
Let $\Cscr(M,N) \in \Cscr^o$ and let $D \in \Cscr(M,N).$
%${i}, {j}: [M]_{{i}, {j}}=[N]_{{i}, {j}}$ so that
%for an instance $D\in\Cscr$ we have for all $u_A\in \Uscr_A, u_B\in\Uscr_B, \xi_W\in \Xi_W$, 
By \eqref{eq:codef}, there exists $u_A^* \in \Uscr_A,u^*_B \in \Uscr_B$ such that
\begin{equation}
	\ell({u}^{*}_A, {u}^{*}_B, \xi_W)\leq \ell(u_A, u_B, \xi_W) \quad \forall u_A,u_B,\xi_W. \label{eq:generalcentral}
\end{equation}
Thus for any $Q\in\Pscr(\Uscr|\Xi)$,
\begin{align*}
&J(Q; D)\\
%&=\sum_{\xi}\Pbb(\xi_A, \xi_B, \xi_W)\sum_{u}\ell(u_A, u_B, \xi_W)Q(u_A, u_B|\xi_A, \xi_B)\\
&\geq \sum_{\xi_{W}} \Pbb(\xi_W|\xi_A, \xi_B)\sum_{\xi_A, \xi_B}\Pbb(\xi_A, \xi_B)\min_{u_A, u_B}\ell(u_A, u_B, \xi_W)\\
&= J(\bar{Q};D),
\end{align*}
where $\bar{Q}$ is as defined in \eqref{eq:ccq}. Since $\bar{Q}\in \Pi,$
we get $J^{**}(D)=J_\Lscr^*(D)$ and the proposition is established.
\end{proof}

Although \eqref{eq:codef} gives a tractable definition of $\Cscr^o$, it is not straightforward to exhaustively enumerate subclasses in $\Cscr^o.$ Hence we will use Lemma~\ref{prop:centralproblems} as an enabling lemma to eliminate some subclasses $\Cscr_{mn} \in \Crsfs.$
Following are two results that accomplish this. 
%We note a simple corner case. If the matrix $M$ (or $N$) is null, the problem trivially admits no centralisation advantage.
\begin{corollary}
	Let either $m=0$ or $n=0$ and let $\Cscr(M,N) \in \Cscr_{mn}$. Then $\Cscr(M, N)$ does not admit a quantum advantage.\label{cor:nonull}
\end{corollary}
\begin{proof}
	If one of the matrices $M$ and $N$ is null, then there exist $u_A^*,u_B^*$ such that 
 \eqref{eq:generalcentral} holds. The rest follows as in Proposition \ref{prop:centralproblems}. 
\end{proof}
\begin{corollary}
Let $\Cscr(M,N)$ be such that $m+n\geq 5$. Then $\Cscr(M, N)$ does not admit quantum advantage.\label{cor:pigeonhole}
\end{corollary}
\begin{proof}
If $m+n\geq 5$, then by the pigeonhole principle, $\Cscr(M, N)\in \Cscr^o$. 
\end{proof}

\subsection{Elimination of other problem classes}
Corollaries \ref{cor:nonull} and \ref{cor:pigeonhole} help eliminate the possibility of a quantum advantage for all $m$-$n$ classes where $m+n\geq 5$ or $\min(m, n)=0$. Thus, out of the $256$ classes in $\Crsfs$, we have eliminated
$\sum_{m+n\geq 5} {^4C_m^4C_n}+\sum_{\min(m, n)=0} {^4C_m^4C_n}=93+31=124$
classes. We now scan through remaining elements in $\Crsfs$, namely,  $\Cscr_{11}, \Cscr_{12}, \Cscr_{21}, \Cscr_{22},\Cscr_{13},\Cscr_{31}$. 

For $i\in\{1, 2\}$, let $-i$ denote the element in $\{1, 2\}\setminus\{i\}$. 
\begin{definition}
We call $V = (M,N)$ and the class $\Cscr(V)$ achiral if $V$ is non-overlapping and $\exists i, j $ such that $ [M]_{ij}=[N]_{-i-j}=-1$. We call a $V$ and the class $\Cscr(V)$ chiral if $V$ is non-overlapping and not achiral.
\end{definition}

\begin{lemma}\label{lem:nooverlapgen}
\begin{enumerate}
	\item  $V$ is overlapping if and only if $V'$ is overlapping for all $V'\in(V;\Omega)$.
	\item  $V$ is achiral if and only if $V'$ is achiral for all $V'\in (V;\Omega)$.
	
\end{enumerate}
\end{lemma}
\begin{proof}
(1) It is easy to see that by inspection all actions in $\Omega$ map an overlapping pair $(M,N)$ to another overlapping pair. Moreover, since $M,N$ are $2\times 2$ matrices, all actions $\Rsf,\Rsf',\Tsf,\Esf$ are involutions, \ie, when applied twice, are equivalent to $\bfI$. In other words if $V' \in (V;\Omega)$, then by a suitable application actions, one can map $V'$ to back to $V$, whereby if $V'$ is overlapping, then so must be $V.$ \\
%Suppose that $V$ is overlapping. Let $i, j$ be such that $[M]_{ij}=[N]_{ij}$. Then, $\Rsf V$ is overlapping since $[\Rsf M]_{-ij}=[M]_{ij}=[N]_{ij}=[\Rsf N]_{-ij}$, $V\Rsf'$ is overlapping since $[\Rsf M]_{i-j}=[M]_{ij}=[N]_{ij}=[\Rsf N]_{i-j}$, $\Tsf V$ is overlapping since $[\Tsf M]_{-i-j}=[M]_{ij}=[N]_{ij}=[\Tsf N]_{-i-j}$ and $\Esf V$ is overlapping since $[\Esf M]_{ij}=[N]_{ij}=[M]_{ij}=[\Esf N]_{ij}$. Thus each of the group operation preserves overlap and $V'$ is coincident for all $V'\in (V;\Omega)$ if $V$ is overlapping. The other direction follows immediately since $V\in (V';\Omega)$ for all $V'\in (V;\Omega)$. \\
(2) This part follows in a similar manner as (1). 
Suppose that $V$ is achiral. Then owing to part (a), every $V'$ in the orbit $(V;\Omega)$ is non-overlapping. We will show that the action $\Rsf$ preserves achirality of $V$; this can be shown for other actions can be proved similarly. 
Let $i, j$ be such that $[M]_{ij}=[N]_{-i-j}=-1$. Then,  $[\Rsf M]_{-ij}=[M]_{ij}=[N]_{-i-j}=[\Rsf N]_{i-j}=-1$, implying that $\Rsf V$ is achiral. Thus, the orbit $(V;\Omega)$ is achiral. Again, using that the actions in $\Omega$ are involutions we get that if any $V'\in (V,\Omega)$ is achiral, then so is $V.$
%$V\Rsf'$ is achiral since $[\Rsf M]_{i-j}=[M]_{ij}=[N]_{-i-j}=[\Rsf N]_{-ij}$, $\Tsf V$ is achiral since $[\Tsf M]_{-i-j}=[M]_{ij}=[N]_{-i-j}=[\Tsf N]_{ij}$ and $\Esf V$ is achiral since $[\Esf M]_{-i-j}=[N]_{-i-j}=[M]_{ij}=[\Esf N]_{ij}$. Thus each of the group operation preserves achirality and $V'$ is achiral for all $V'\in (V;\Omega)$ if $V$ is achiral. The other direction follows immediately since $V\in (V';\Omega)$ for all $V'\in (V;\Omega)$.\\
\end{proof}

Our elimination procedure for $\Cscr_{mn}$ can be described as follows. We define $\Cscr_{mn}^o=\Cscr_{mn}\cap \Cscr^o$ as the collection of all overlapping $m$-$n$ classes. Observe that 
\begin{equation}\label{eq:overlap}
	|\Cscr_{mn}^o| = \comb{4}{m} \times (\sum_{k=1}^{n}\comb{m}{k}\comb{4-m}{n-k} ),
\end{equation}
since we have $\comb{4}{m}$ choices for a `$-1$' in $M$ following which we have $\comb{m}{k}\comb{4-m}{n-k}$ choices for $k$ overlapping $-1$'s in $N$ and $\comb{4-m}{n-k}$ for the remaining $(n-k)$ nonoverlapping $-1$'s.
We then explicitly specify a chiral $V_c=(M_c, N_c)$ and an achiral $V_a=(M_a, N_a)$ and define $\Cscr_{mn}^c:=\Cscr((V_c;\Omega))\cap \Cscr_{mn}$, $\Cscr_{mn}^a:=\Cscr((V_a;\Omega))\cap\Cscr_{mn}$. Following Lemma \ref{lem:nooverlapgen}, such a specification ensures that all classes in $\Cscr_{mn}^c$ are chiral and all those in $\Cscr_{mn}^a$ are achiral so that $\Cscr_{mn}^o,\Cscr_{mn}^a,\Cscr_{mn}^c$ are mutually disjoint.
%$$\Cscr_{mn}^a\cap \Cscr^o_{mn}=\Cscr_{mn}^c\cap \Cscr^o_{mn}=\Cscr_{mn}^a\cap \Cscr^c_{mn}=\phi.$$
We then establish that our choice of $V_c$, $V_a$ ensures that $\Cscr^o_{mn}, \Cscr^a_{mn}$ and $\Cscr^c_{mn}$ exhaust $\Cscr_{mn},$ whereby these constitute a partition of $\Cscr_{mn}$.
We then examine $\Cscr(V_c)$ and $\Cscr(V_a)$ and eliminate those that do not admit a quantum advantage.
% do not admit the no-signalling, and hence the quantum advantage which in turn eliminates $\Cscr_{mn}^c$ and $\Cscr_{mn}^a$. Since $\Cscr_{mn}^o\subset \Cscr^o$, its elimination is automatic from proposition \ref{prop:centralproblems}. This completes the elimination procedure for the entire set $\Cscr_{mn}$. 

\subsubsection{Elimination of 1-1 problem class $\Cscr_{11}$}
Consider $\Cscr_{11}$ and notice $|\Cscr_{11}|=\prescript{4}{}C_1\prescript{4}{}C_1=16$. Define 
$\Cscr_{11}^o=\Cscr_{11}\cap\Cscr^o.$
Define the following achiral pair $V_a=(M_a, N_a)$,
\begin{equation}\label{eq:11ma}
 M_a:=\begin{pmatrix}
-1&0\\
0&0
\end{pmatrix},\quad N_a:=\begin{pmatrix}
0&0\\
0&-1
\end{pmatrix},
\end{equation}
and let $\Cscr_{11}^a:=  \Cscr((V^a;\Omega))\cap \Cscr_{11}$. Observe that,
\begin{equation}\label{eq:c11adef}
\Cscr_{11}^a=\{\Cscr(V_a), \Cscr(\Rsf V_a),
 \Cscr(V_a\Rsf), \Cscr(\Rsf V_a\Rsf)\},
\end{equation}
so that $|\Cscr_{11}^a|=4$. Now take the chiral pair $V_c=(M_c, N_c),$
\begin{equation}\label{eq:11mc}
M_c:=\begin{pmatrix}
-1&0\\
0&0
\end{pmatrix}, N_c:=\begin{pmatrix}
0&-1\\
0&0
\end{pmatrix},
\end{equation}
and let $\Cscr_{11}^c := \Cscr((V^c;\Omega))\cap \Cscr_{11}$. It is easy to verify that
\begin{align}\label{eq:c11cdef}
\Cscr_{11}^c=&\{\Cscr(V_c), \Cscr(\Rsf V_c), \Cscr(V_c\Rsf), \Cscr(\Rsf V_c\Rsf), \Cscr(\Tsf V_c),\nonumber\\
&\Cscr(\Rsf \Tsf V_c), \Cscr(\Tsf V_c\Rsf), \Cscr(\Rsf \Tsf V_c\Rsf)\},
\end{align}
whereby $|\Cscr_{11}^c|=8 = |\Cscr_{11}|-|\Cscr_{11}^o|-|\Cscr_{11}^a|.$ Thus $\Cscr_{11}^c,\Cscr_{11}^o,\Cscr_{11}^a$ is a partition of $\Cscr_{11}.$
The following proposition eliminates $\Cscr_{11}$ by elimination of each of the elements in this partition.
\begin{proposition} \label{prop:11eliminate} $\Cscr_{11}$ does not admit a quantum advantage.
%\begin{enumerate}
%	\item  
%	\item  $\Cscr(M_a,N_a)$ does not admit a no-signalling advantage for $(M_a, N_a)$ in \eqref{eq:11ma}. Hence $\Cscr_{11}^a$ does not admit quantum advantage since $\Cscr_{11}^a =\Cscr((V_a, \Omega))$.
%	\item  $\Cscr(V_c)$ does not admit a no-signalling advantage for $V_c:=(M_c, N_c)$ in \eqref{eq:11mc}. Hence $\Cscr_{11}^c$ does not admit quantum advantage since $\Cscr_{11}^c=\Cscr((V_c;\Omega))$.
%	\item  $\Cscr_{11}$ does not admit a quantum advantage.
%\end{enumerate}
\end{proposition}

\begin{proof} $\Cscr_{11}^o$ does not admit a quantum advantage 
since $\Cscr_{11}^o\subset \Cscr^o$. 
We now show the same for $\Cscr_{11}^a.$ For an instance $D\in\Cscr(M_a, N_a)$, for $(M_a,N_a)$ as defined in \eqref{eq:11ma}, note that  $\ell(u_A^0,u_B^0,0)=-1$, $\ell(u_A^1, u_B^1, 1)=-\chi$ and $\ell(.)\equiv 0$ otherwise. Now consider deterministic policies $\hat{\gamma},\overline{\gamma}:$ $$ \hat{\gamma}_A(\xi_A)\equiv u_A^0, \hat{\gamma}_B(\xi_B)\equiv u_B^0 \aur \overline{\gamma}_A(\xi_A)\equiv u_A^1,\overline{\gamma}_B(\xi_B)\equiv u_B^1.$$
It is easy to evaluate,
\[ J(\pi_{\hat{\gamma}};D) = - \Pbb(\xi_w=0), \quad J(\pi_{\overline{\gamma}};D)=-\chi \Pbb(\xi_W=1) \]
%\begin{align*}
%J(\pi_{\hat{\gamma}};D)&=-\sum_{\xi_A, \xi_B} \Pbb(\xi_A,\xi_B,0), \ J(\pi_{\overline{\gamma}};D)=-\chi\sum_{\xi_A, \xi_B} \Pbb(\xi_A,\xi_B,1)
%\end{align*}
Now consider for a no-signalling vertex $Q^{\alpha\beta\delta}\in\Nscr\Sscr$ (recall \eqref{eq:nonlocalvertex}), 
\begin{align*}
& J(Q^{\alpha\beta\delta}; D)
=-\sum_{\xi_A, \xi_B}(\Pbb(\xi_A,\xi_B,0)Q^{\alpha\beta\delta}(u_A^0,u_B^0|\xi_A,\xi_B)\\
&\qquad\qquad +\chi\Pbb(\xi_A,\xi_B,1)Q^{\alpha\beta\delta}(u_A^0,u_B^1|\xi_A,\xi_B))\\
&=-\frac{1}{2}\sum_{\xi_A, \xi_B}(\Pbb(\xi_A, \xi_B,0)(\sim\xi_A\cdot\xi_B\oplus\alpha\cdot\xi_A\oplus\beta\cdot\xi_B\oplus\delta)
\\
&\qquad\qquad+\chi\Pbb(\xi_A, \xi_B, 1)(\sim\xi_A\cdot\xi_B\oplus\alpha\cdot\xi_A\oplus\beta\cdot\xi_B\oplus\delta))\\
%&\geq -\frac{1}{2}\sum_{\xi_A, \xi_B}\left(\Pbb(\xi_A,\xi_B,0)
%+\chi\Pbb(\xi_A,\xi_B,1)\right)\\
&\geq \frac{1}{2}\left(J(\pi_{\hat{\gamma}};D)+J(\pi_{\overline{\gamma}};D)\right),
\end{align*}
where in the last step we have used that the terms multiplying the probabilities are nonnegative. Since the RHS is independent of the no-signalling vertex, the cost of every no-signalling vertex is bounder below by the cost of a deterministic policy
%\begin{equation}
$J(Q^{\alpha\beta\delta}; D)\geq \min(J(\pi_{\hat{\gamma}};D), J(\pi_{\overline{\gamma}};D)). $
%\end{equation}
Since the instance $D$ was arbitrary, this establishes that $\Cscr_{11}^a$ does not admit a no-signalling and hence quantum advantage.

We follow a similar line of arguments for $\Cscr_{11}^c.$ For an instance $D\in\Cscr(M_c, N_c)$, we have $\ell(u_A^0,u_B^0,0)=-1$, $\ell(u_A^0, u_B^1, 1)=-\chi$ and $\ell(.)\equiv 0$ otherwise. Now consider deterministic policies $\hat{\gamma}$: $\hat{\gamma}_A(\xi_A)\equiv u_A^0, \hat{\gamma}_B(\xi_B)\equiv u_B^0$ and $\overline{\gamma}$: $\overline{\gamma}_A(\xi_A)\equiv u_A^0,\overline{\gamma}_B(\xi_B)\equiv u_B^1$
It is straightforward to evaluate
\begin{align*}
J(\pi_{\hat{\gamma}};D)&=-\Pbb(\xi_{W}=0)\quad 
J(\pi_{\overline{\gamma}};D)=-\chi \Pbb(\xi_w=1).
\end{align*}
Now for any no-signalling vertex $Q^{\alpha\beta\delta}\in\Nscr\Sscr$, we again have,
\begin{align*}
J(Q^{\alpha\beta\delta}; D)
%&=\sum_{\xi, u}\Pbb(\xi_A, \xi_B, \xi_B)\ell(u_A, u_B, \xi_W)Q^{\alpha\beta\delta}(u_A, u_B|\xi_A, \xi_B)\\
%&=-\sum_{\xi_A, \xi_B}(\Pbb(\xi_A,\xi_B,0)Q^{\alpha\beta\delta}(u_A^0,u_B^0|\xi_A,\xi_B)\\
%&\qquad\qquad\qquad+\chi\Pbb(\xi_A,\xi_B,1)Q^{\alpha\beta\delta}(u_A^0,u_B^1|\xi_A,\xi_B))\\
%&=-\frac{1}{2}\sum_{\xi_A, \xi_B}(\Pbb(\xi_A, \xi_B,0)(\sim\xi_A\cdot\xi_B\oplus\alpha\cdot\xi_A\oplus\beta\cdot\xi_B\oplus\delta)\\
%&\qquad\qquad +\chi\Pbb(\xi_A, \xi_B, 1)(\xi_A\cdot\xi_B\oplus\alpha\cdot\xi_A\oplus\beta\cdot\xi_B\oplus\delta))\\
%&\geq -\frac{1}{2}\sum_{\xi_A, \xi_B}\left(\Pbb(\xi_A,\xi_B,0)
%+\chi\Pbb(\xi_A,\xi_B,1)\right)\\
&\geq \frac{1}{2}\left(J(\pi_{\hat{\gamma}};D)+J(\pi_{\overline{\gamma}};D)\right),
\end{align*}
whereby $J^*_{\Nscr\Sscr}(D)=J^*_\Lscr(D),$ and that $\Cscr_{11}^c$ does not admit a quantum advantage.
%Thus, the cost of every no-signalling vertex is bounder below by the cost of a deterministic policy:
%\begin{equation}
%J(Q^{\alpha\beta\delta}; D)\geq \min(J(\pi_{\hat{\gamma}};D), J(\pi_{\overline{\gamma}};D))
%\end{equation}
This establishes the proposition. 
\end{proof}

\subsubsection{Elimination of 1-3 and 3-1 problem classes $\Cscr_{13}$ and $\Cscr_{31}$}
Now consider the set of $1$-$3$ class $\Cscr_{13}$. We argue that it does not admit a quantum advantage. Since $\Cscr_{31}$ is can generated by from $\Cscr_{13}$ by action $\Esf$, thanks to Proposition~\ref{prop:exchange} we need not discuss $\Cscr_{31}$ separately.
 Again, define 
$\Cscr_{13}^o=\Cscr_{13}\cap\Cscr^o$, 
%To see that $|\Cscr_{13}^o|=12$, we have $4$ choices for a $-1$ entry in $M$ following which three possible choices for the three entries in $N$ ensure overlap.
define the achiral pair $V_a=(M_a, N_a)$ as below,
\begin{equation}\label{eq:13ma}
M_a:=\begin{pmatrix}
-1&0\\
0&0
\end{pmatrix}, N_a:=\begin{pmatrix}
0&-1\\
-1&-1
\end{pmatrix},
\end{equation}
and let, $\Cscr_{13}^a:=\Cscr( (V^a, \Omega))\cap \Cscr_{13}$ and note that
\begin{equation}\label{eq:13ca}
\Cscr_{13}^a=\{\Cscr(V_a), \Cscr(\Rsf V_a), \Cscr(V_a\Rsf), \Cscr(\Rsf V_a\Rsf)\}.
\end{equation}
Notice that we do not have a chiral class in $\Cscr_{13}$, since in such a class the matrix $N$ would have $3$ entries as $-1$, none of which overlap with the $-1$ entry in $M$, and must not be of the form in \eqref{eq:13ma}.
%. If $V=(M, N)$ is chiral and $\Cscr(V)\in\Cscr_{13}$, then for $i, j$ corresponding to $M_{ij}=-1$, chirality requires no overlap so that $N_{ij}=0$, and absence of achirality, $N_{-i-j}=0$ so that the number of non-zero entries in $N$ is $\leq 2<3$. 
%So we let $\Cscr_{13}^c=\phi$.
Notice that $|\Cscr_{13}|=\prescript{4}{}C_1\prescript{4}{}C_3=16$, $|\Cscr^o_{13}|=12$ from \eqref{eq:overlap}, and   $|\Cscr_{13}^a|=4=|\Cscr_{13}|-|\Cscr_{13}^o|$ so that $\Cscr^a_{13}$ and $\Cscr_{13}^o$ indeed themselves partition the set of all $1$-$3$ classes.
%$$\Cscr^o_{13}\cap\Cscr^a_{13}=\phi.$$
%$$|\Cscr^o_{13}|+|\Cscr^a_{13}|=|\Cscr_{13}|=16.$$
\begin{proposition}
$\Cscr_{13}$ and $\Cscr_{31}$ do not admit a quantum advantage.
\label{prop:eliminatecomp}
\end{proposition}
\begin{proof}
$\Cscr^o_{13}$ does not admit quantum advantage since $\Cscr_{13}^o\subset \Cscr^o$. Now to show the same for $\Cscr^a_{13}$, consider an instance $D\in \Cscr(M_a, N_a)$ for $V_a=(M_a, N_a)$ as in \eqref{eq:13ma}. Recall from \eqref{eq:localvertex} that $\pi^{abcd}$ denotes a local deterministic strategy for all binary $a, b, c, d$.  For any no-signalling vertex $Q^{\alpha\beta\delta}$, we claim that
\begin{equation}\label{eq:13nosigdec}
J(Q^{\alpha\beta\delta}; D)=\frac{1}{2}(J(\pi^{xyzw}; D)+J(\pi^{11ab};D))
\end{equation}
where $\pi^{xyzw}$ and $\pi^{11ab}$, as defined in \eqref{eq:localvertex}, are vertices of the local polytope specified by the Boolean variables $x,y,z,w,a,b\in\{0, 1\}$, which in turn are given  in terms of $\alpha, \beta$ and $\delta$ as,
\begin{align}\label{eq:x}
x&=(\sim\beta\cdot\sim\delta)\vee(\beta\cdot\alpha\cdot\delta)\\
y&=\sim\alpha\cdot\sim\delta\cdot\beta\label{eq:y}\\
z&=\delta\vee(\sim\delta\cdot\alpha\cdot\beta)\label{eq:z}\\
w&=(\alpha\cdot(\beta\oplus\delta))\vee(\sim\alpha\cdot\delta)\label{eq:w}\\
a&=\sim\beta\vee(\beta\cdot\sim\alpha\cdot\sim\delta)\label{eq:a}\\
b&=(\sim\alpha\cdot(\beta\vee\delta))\vee(\alpha\cdot(\sim\beta\oplus\delta))\label{eq:b}.
\end{align}
To establish this claim, notice for $D\in\Cscr(V_a)$, and any $Q\in \Pscr(\Uscr|\Xi)$,
\begin{align}
J(Q; D)&=-\sum_{\xi_A, \xi_b}\left (\Pbb(\xi_A, \xi_B, 0)Q(u_A^0,u_B^0|\xi_A, \xi_B)\nonumber\right.\\
&\left.-\chi\Pbb(\xi_A, \xi_B, 1)(1-Q(u_A^0,u_B^0|\xi_A, \xi_B))\right ).\label{eq:1-3nscost} 
%J&(Q^{\alpha\beta\delta}; D)=-\sum_{\xi_A, \xi_b}\Pbb(\xi_A, \xi_B, 0)Q^{\alpha\beta\delta}(u_A^0,u_B^0|\xi_A, \xi_B)\nonumber\\
%&-\chi\Pbb(\xi_A, \xi_B, 1)(Q^{\alpha\beta\delta}(u_A^0,u_B^1|\xi_A, \xi_B)\nonumber\\
%&+Q^{\alpha\beta\delta}(u_A^1,u_B^0|\xi_A, \xi_B)+Q^{\alpha\beta\delta}(u_A^1,u_B^1|\xi_A, \xi_B)),\label{eq:1-3nscost} \\
%J&(\pi^{xyzw}; D)\nonumber\\
%&=-\sum_{\xi_A, \xi_b}\Pbb(\xi_A, \xi_B, 0)\pi^{xyzw}(u_A^0, u_B^0|\xi_A, \xi_B)\nonumber\\
%&-\chi\Pbb(\xi_A, \xi_B, 1)(\pi^{xyzw}(u_A^0, u_B^1|\xi_A, \xi_B)\nonumber\\
%&+\pi^{xyzw}(u_A^1, u_B^0|\xi_A, \xi_B)+\pi^{xyzw}( u_A^1, u_B^1|\xi_A, \xi_B)),
%\label{eq:1-3detcost1} \\
%J&(\pi^{11ab}; D)\nonumber\\
%&=-\sum_{\xi_A, \xi_b}\Pbb(\xi_A, \xi_B, 0)\pi^{11ab}(u_A^0, u_B^0|\xi_A, \xi_B)\nonumber\\
%&-\chi\Pbb(\xi_A, \xi_B, 1)(\pi^{11ab}(u_A^0, u_B^1|\xi_A, \xi_B)\nonumber\\
%&+\pi^{11ab}(u_A^1,0|\xi_A, \xi_B)+\pi^{11ab}(u_A^1, u_B^1|\xi_A, \xi_B)).\label{eq:1-3detcost2}
\end{align}
To establish \eqref{eq:13nosigdec},  we  show that for all $\alpha, \beta, \delta, \xi_A, \xi_B \in\{0, 1\}$, and with $x,y,z,w,a,b$ as specified by \eqref{eq:x}-\eqref{eq:b},
\begin{align}\label{eq:pyclaim1}
&Q^{\alpha\beta\delta}(u_A^0, u_B^0|\xi_A, \xi_B)\nonumber\\
&=\frac{1}{2}(\pi^{xyzw}(u_A^0, u_B^0|\xi_A, \xi_B)+\pi^{11ab}(u_A^0, u_B^0|\xi_A, \xi_B)),
\end{align}
%and
%\begin{align}\label{eq:pyclaim2}
%&Q^{\alpha\beta\delta}(u_A^0, u_B^1|\xi_A, \xi_B)+Q^{\alpha\beta\delta}(u_A^1, u_B^0|\xi_A, \xi_B)\nonumber \\
%&+Q^{\alpha\beta\delta}(u_A^1, u_B^1|\xi_A, \xi_B)\nonumber\\
%&=\frac{1}{2}(\pi^{xyzw}(u_A^0, u_B^1|\xi_A, \xi_B)+\pi^{11ab}(u_A^0, u_B^1|\xi_A, \xi_B))\nonumber\\
%&+\pi^{xyzw}(u_A^1, u_B^0|\xi_A, \xi_B)+\pi^{11ab}(u_A^1, u_B^0|\xi_A, \xi_B)\nonumber\\
%&+\pi^{xyzw}(u_A^1, u_B^1|\xi_A, \xi_B)+\pi^{11ab}(u_A^1, u_B^1|\xi_A, \xi_B)).
%\end{align}
so that \eqref{eq:13nosigdec} now follows from \eqref{eq:pyclaim1}.

The validity of \eqref{eq:pyclaim1} can be done through straightforward computation; due to the large number of variables involved, we relegate this to a Python notebook~\cite{c13pylink} provided in the supplementary material. This establishes our claim
\eqref{eq:13nosigdec}. 

We have thus established that for each no-signalling $Q^{\alpha\beta\delta}$ policy, there exists a policy in $\pi \in \Lscr$ such that $J(Q^{\alpha\beta\delta};D)\geq J(\pi;D),$ whereby establishing that $J^*_{\Nscr\Sscr}(D) \geq J^*_\Lscr(D).$ Since $D$ is arbitrary, there is no quantum advantage in $\Cscr_{13},$ and from Proposition~\ref{prop:exchange}, none in $\Cscr_{31}.$
\end{proof}
 
\subsection{$\Cscr_{12}$, $\Cscr_{21}$ and the \textit{1/2-CAC} problem class}
We now come to the $1$-$2$ class 
$\Cscr_{12}$; we will quickly address $\Cscr_{21}$ at the end of this subsection. Define $\Cscr_{12}^o=\Cscr_{12}\cap\Cscr^o,$
%we have  Define 
%It is easy to obtain $|\Cscr^o_{12}|=12$ combinatorially.
%since we have $4$ choices for a non-zero entry in $M$ following which two entries of $N$ can be placed in $3$ different ways.
the achiral pair $V_a=(M_a, N_a)$, 
\begin{equation}\label{eq:12ma}
 M_a:=\begin{pmatrix}
-1&0\\
0&0
\end{pmatrix}, N_a:=\begin{pmatrix}
0&-1\\
0&-1
\end{pmatrix},
\end{equation}
and the chiral pair $V_c=(M_c, N_c)$ 
\begin{equation}\label{eq:12mc}
	M_c:=\begin{pmatrix}
		-1&0\\
		0&0
	\end{pmatrix}, N_c:=\begin{pmatrix}
		0&-1\\
		-1&0
	\end{pmatrix}.
\end{equation}
Note that the chiral pair is $\half$-CAC form. Let $\Cscr_{12}^a:=\Cscr( (V^a, \Omega))\cap \Cscr_{12}$ and $\Cscr_{12}^c:=\Cscr( (V_c, \Omega))\cap \Cscr_{12}$. 
Note that
\begin{align}\label{eq:c12adef}
 \Cscr_{12}^a=&\{\Cscr(V_a), \Cscr(\Rsf V_a),
 \Cscr(V_a\Rsf), \Cscr(\Rsf V_a\Rsf), \Cscr(\Tsf V_a), \nonumber\\
 &\quad \Cscr(\Rsf\Tsf V_a),
 \Cscr(\Tsf V_a\Rsf), \Cscr(\Rsf \Tsf V_a\Rsf)\}, \\
\Cscr_{12}^c =&\{\Cscr(V_c), \Cscr(\Rsf V_c), \Cscr(V_c\Rsf), \Cscr(\Rsf V_c\Rsf)\}.  \label{eq:c12cdef}
\end{align}
Further, notice that $|\Cscr_{12}|=\prescript{4}{}C_1\prescript{4}{}C_2=24$, 
 $|\Cscr_{12}^o|+|\Cscr_{12}^a|+|\Cscr_{12}^c|=\Cscr_{12}=24$ so $\Cscr_{12}^o$, $\Cscr_{12}^a$ and $\Cscr_{12}^c$ partition the set $\Cscr_{12}$.
%$$\Cscr_{12}^c\cap \Cscr^o_{12}=\Cscr_{12}^a\cap \Cscr^o_{12}=\Cscr_{12}^c\cap \Cscr^a_{12}=\phi.$$
%$$|\Cscr^o_{12}|+|\Cscr^c_{12}|+|\Cscr^a_{12}|=|\Cscr_{12}|=16.$$
We eliminate all 1-2 classes not in the orbit of the $\half$-CAC class (\ie, not in the orbit of the chiral pair $(M_c,N_c)$) in the following proposition.
\begin{proposition}\label{prop:1-2Celiminate}
1) $\Cscr_{12}^o$ does not admit a quantum advantage.\\
2) $\Cscr_{12}^a$ does not admit quantum advantage.
\end{proposition}
\begin{proof}
1) Immediate from $\Cscr_{12}^o\subset \Cscr^o$.\\
2) For an instance $D\in\Cscr(V_a)$, Consider two deterministic policies $\hat\gamma$ and $\overline \gamma$, and the corresponding costs:
\begin{align}
\hat{\gamma}_A(\xi_A)&\equiv u_A^0, \hat{\gamma}_B(\xi_B)\equiv u_B^0; J(\pi_{\hat{\gamma}};D)=-\Pbb(\xi_W=0)\nonumber \\
\overline{\gamma}_A(\xi_A)&\equiv 1,\overline{\gamma}_B(\xi_B)\equiv 1;J(\pi_{\overline{\gamma}};D)=-\chi \Pbb(\xi_W=1).\label{eq:elidet}
\end{align}
Now consider for a no-signalling vertex $Q^{\alpha\beta\delta}\in\Nscr\Sscr$ and recall \eqref{eq:nonlocalvertex} to express 
\begin{align*}
&J(Q^{\alpha\beta\delta}; D)=-\sum_{\xi_A, \xi_B} \Pbb(\xi_A,\xi_B,0)Q^{\alpha\beta\delta}(u_A^0,u_B^0|\xi_A,\xi_B)+\\
& \chi\Pbb(\xi_A,\xi_B,1)\sum_{u_A} Q^{\alpha\beta\delta}(u_A, u_B^1|\xi_A, \xi_B)\\
&=\frac{-1}{2}\sum_{\xi_A, \xi_B}(\Pbb(\xi_A, \xi_B,0)(\sim\xi_A\cdot\xi_B\oplus\alpha\cdot\xi_A\oplus\beta\cdot\xi_B\oplus\delta)\\
&\qquad\qquad+\chi\Pbb(\xi_A, \xi_B, 1)(\xi_A\cdot\xi_B\oplus\alpha\cdot\xi_A\oplus\beta\cdot\xi_B\oplus\delta\\
&\qquad\qquad + \sim\xi_A\cdot\xi_B\oplus\alpha\cdot\xi_A\oplus\beta\cdot\xi_B\oplus\delta)),\\
%&\geq -\frac{1}{2}\sum_{\xi_A, \xi_B}(\Pbb(\xi_A,\xi_B,0)
%+\chi\Pbb(\xi_A,\xi_B,1))\\
&\geq \frac{1}{2}(J(\pi_{\hat{\gamma}};D)+J(\pi_{\overline{\gamma}};D)).
\end{align*}
In the last inequality we have again used the nonnegativity of the terms multiplying the probabilities. 
Thus, the cost of every no-signalling policy is bounded below by the cost of a deterministic policy in $\Lscr$. Arguing as in Proposition~\ref{prop:11eliminate}, we see that there is no no-signalling advantage and quantum advantage within in $\Cscr^a_{12}$. This establishes the proposition.
\end{proof}
Now notice that $\Cscr_{21}=\Esf\Cscr_{12}$. Thus define $\Cscr_{21}^o=\Esf\Cscr_{12}^o, \Cscr_{21}^a=\Esf\Cscr_{12}^a$ and $\Cscr_{21}^c=\Esf\Cscr_{12}^c$, and the 2-1 class partitions into $\Cscr_{21}^o$, $\Cscr_{21}^a$ and $\Cscr_{21}^c$.
$\Cscr_{21}^c$ here lies within the orbit of the $\half$-CAC class, and the elimination of the other two $\Cscr_{21}^o$ and $\Cscr_{21}^a$ follows from Proposition \ref{prop:1-2Celiminate} and Proposition~\ref{prop:11eliminate}. This subsection thus eliminates all 1-2 and 2-1 classes that do not lie in the orbit of $\half$-CAC class.

\subsection{$\Cscr_{22}$ and the CAC problem class}
Ultimately, we attend the set of 2-2 classes  $\Cscr_{22}$. Define 
$\Cscr_{22}^o=\Cscr_{22}\cap\Cscr^o$, 
the chiral pair $V_c=(M_c, N_c)$ and the achiral pair $V_a=(M_a, N_a)$
\begin{equation}\label{eq:22mc}
M_c:=\begin{pmatrix}
-1&0\\
0&-1
\end{pmatrix}, N_c:=\begin{pmatrix}
0&-1\\
-1&0
\end{pmatrix}.
\end{equation}
\begin{equation}\label{eq:22ma}
 M_a:=\begin{pmatrix}
-1&0\\
-1&0
\end{pmatrix}, N_a:=\begin{pmatrix}
0&-1\\
0&-1
\end{pmatrix}.
\end{equation}
The chiral pair is in the CAC form. 
Define $\Cscr_{22}^c:= \Cscr((V^c, \Omega))\cap \Cscr_{22}$ and $ \Cscr_{22}^a:=\Cscr( (V^a, \Omega))\cap\Cscr_{22}$. Notice that,
\begin{align}
\Cscr_{22}^c&=\{\Cscr(V_c), \Cscr(\Rsf V_c)\},\label{eq:c22cdef}\\
\Cscr_{22}^a&=\{\Cscr(V_a),
 \Cscr(V_a\Rsf), \Cscr(\Tsf V_a), \Cscr(\Rsf \Tsf V_a)\}. \label{eq:c22adef}
\end{align} 
It is easy to check by inspection, and using \eqref{eq:overlap}, $|\Cscr^o_{22}|+|\Cscr^a_{22}|+|\Cscr^c_{22}|=|\Cscr_{22}|$
so that $\Cscr^o_{22}$, $\Cscr^a_{22}$ and $\Cscr^c_{22}$ partition $\Cscr_{22}$. Consequently, the sets $\Cscr^o_{22}$ and $\Cscr^a_{22}$ capture all 2-2 classes which are outside the orbit of CAC, and we eliminate these sets in the following proposition.
\begin{proposition}\label{prop:2-2Celiminate}
1) $\Cscr_{22}^o$ does not admit a quantum advantage. \\
2) $\Cscr_{22}^a$ does not admit quantum advantage.
\end{proposition}
\begin{proof}
1) Immediate since $\Cscr_{22}^o\subset \Cscr^o$.\\
2) %For an instance $D\in \Cscr(V_a)$, consider deterministic policies $\hat{\gamma}:\hat{\gamma}_A(\xi_A)\equiv u_A^0, \hat{\gamma}_B(\xi_B)\equiv u_B^0$ and $\overline{\gamma}:\overline{\gamma}_A(\xi_A)\equiv u_A^1,\overline{\gamma}_B(\xi_B)\equiv u_B^1$. It is straightforward to evaluate the cost of $D$ for each of these policies:
%\begin{align*}
%J(\pi_{\hat{\gamma}};D)&=-\sum_{\xi_A, \xi_B} \Pbb(\xi_A,\xi_B,0),\\
%J(\pi_{\overline{\gamma}};D)&=-\chi\sum_{\xi_A, \xi_B} \Pbb(\xi_A,\xi_B,1).
%\end{align*}
Let $D$ be an instance in $\Cscr_{22}^a$ and let $\hat \gamma$ and $\overline\gamma$ be as defined in \eqref{eq:elidet}. Consider any no-signalling policy $Q^{\alpha\beta\delta}\in\Nscr\Sscr$ and notice,
\begin{align*}
&J(Q^{\alpha\beta\delta}; D)=-\sum_{\xi_A, \xi_B}(\Pbb(\xi_A,\xi_B,0)(Q^{\alpha\beta\delta}(u_A^0,u_B^0|\xi_A,\xi_B))\\
&\qquad +\chi\Pbb(\xi_A,\xi_B,1)(\sum_{u_A}Q^{\alpha\beta\delta}(u_A^0,u_B^1|\xi_A,\xi_B)))	
\end{align*}
Using \eqref{eq:nonlocalvertex},
\begin{align*}
J(Q^{\alpha\beta\delta}; D)&=-\frac{1}{2}\sum_{\xi_A, \xi_B}(\Pbb(\xi_A, \xi_B,0)+\chi\Pbb(\xi_A, \xi_B,1))\times\\
&(\sim\xi_A\cdot\xi_B\oplus\alpha\cdot\xi_A\oplus\beta\cdot\xi_B\oplus\delta \\
&\qquad \qquad+ \xi_A\cdot\xi_B\oplus\alpha\cdot\xi_A\oplus\beta\cdot\xi_B\oplus\delta)\\
%&= -\frac{1}{2}\sum_{\xi_A, \xi_B}(\Pbb(\xi_A,\xi_B,0)
%+\chi\Pbb(\xi_A,\xi_B,1))\\
&=\frac{1}{2}(J(\pi_{\hat{\gamma}};D)+J(\pi_{\overline{\gamma}};D)).
\end{align*}
Arguing as in Proposition~\ref{prop:11eliminate}, $J^*_{\Nscr\Sscr}(D)=J^*_\Lscr(D) $, 
 and the no-signalling and quantum advantages are absent in $\Cscr_{22}^a$. This establishes the proposition.
\end{proof}

\section{Proof of Theorem~\ref{thm:main}: Quantum Advantage in $\half$-CAC}\label{sec:1/2cacpoc}
We now have all but one ingredient to prove Theorem~\ref{thm:main}. We have shown that all classes not in the orbit of the CAC and $\half$-CAC class do not admit a quantum advantage. That CAC admits a quantum advantage was shown in \cite{deshpande2022quantum}. We now show this for $\half$-CAC. Consider a problem instance $D$ in $\half$-CAC class with the specification $D=(M_c, N_c, \Pbb, \Uscr_A, \Uscr_B, \chi)$ where $M_c, N_c$ are as in \eqref{eq:12mc}, and $\Pbb\in\Pscr(\Xi)$ is such that 
\begin{equation}\label{eq:1/2cacprior}
\Pbb(\xi_A, \xi_B, \xi_W)=\begin{cases} 0.2& \xi= (0,0,1), (0,1,1), (1,0,1)\\
0.4 & \xi=(1,1,0)\\
0 &\text{otherwise}
\end{cases}
\end{equation}
and $\chi=2$. It is straightforward to sift through all 16 deterministic strategies in $\Pi$. We state at optimal policy and the optimal local cost here, $\gamma_A^*\equiv u_A^0$, $\gamma_B^*\equiv u_B^1$ and $J_\Lscr^*(D)=-6/5$, and justify this statement in Lemma \ref{lem:1/2cacdemodet} in the appendix.

We now specify a quantum strategy $Q$ that achieves a  lower cost. We consider two dimensional Hilbert spaces $\Hscr_A$, $\Hscr_B$ and a four dimensional $\Hscr=\Hscr_A\otimes \Hscr_B$. Let $\{\ket{z^+_i}, \ket{z^-_i}\}$ be an orthonormal basis of $\Hscr_i$ for $i\in\{A, B\}$. We work with a Euclidean  representation $\ket{z^+_i}\equiv (1, 0)\t; \ket{z^-_i}\equiv (0, 1)\t$. $\Hscr$ is then spanned by $\{\ket{z^+_A}\otimes \ket{z^+_B}, \ket{z^+_A}\otimes \ket{z^-_B}, \ket{z^-_A}\otimes \ket{z^+_B}, \ket{z^-_A}\otimes \ket{z^-_B}\}$ and enumerated in that order. Let $\rho_{AB}$ be the following density operator on $\Hscr$,
$$\rho_{AB}=\begin{pmatrix}
1/4&0&0&\sqrt{3}/{4}\\
0&0&0&0\\
0&0&0&0\\
{\sqrt{3}}/{4}&0&0&{3}/{4}
\end{pmatrix}.$$
It is clear that $\rho_{AB}$ satisfies $\rho_{AB}^\dagger=\rho_{AB}, \Tr( \rho_{AB})=1$:
Next, we specify the projection operators $P_{u_A}^{(A)}(\xi_A)\in\Bscr(\Hscr_A), P_{u_B}^{(B)}(\xi_B)\in\Bscr(\Hscr_B)$:
\begin{align*}
P_{u_A^0}^{(A)}(0)&=\begin{pmatrix} 1&0\\0&0 \end{pmatrix}; P_{u_A^1}^{(A)}(0)=\Ibf - P_{u_A^0}^{(A)}(0),\\
P_{u_A^0}^{(A)}(1)&=\frac{1}{2} \begin{pmatrix} 1 & e^{-\iota \frac{\pi}{3}} \\e^{\iota \frac{\pi}{3}} & 1 \end{pmatrix}; P_{u_A^1}^{(A)}(1)=\Ibf-P_{u_A^1}^{(A)}(1),\\
P_{u_B^0}^{(B)}(0)&=\frac{1}{4}\begin{pmatrix} 2-\sqrt{3} & e^{-\iota \frac{\pi}{3}}\\ e^{\iota \frac{\pi}{3}} & 2+\sqrt{3} \end{pmatrix}; P_{u_B^1}^{(B)}(0)=\Ibf-P_{u_B^0}^{(B)}(0),\\
P_{u_B^0}^{(B)}(1)&=\frac{1}{4}\begin{pmatrix} 2-\sqrt{3} & e^{\iota \frac{\pi}{3}}\\ e^{-\iota \frac{\pi}{3}} & 2+\sqrt{3} \end{pmatrix}; P_{u_B^1}^{(B)}(1)=\Ibf-P_{u_B^0}^{(B)}(0).
\end{align*} $P_{u_A^0}^{(A)}(0)$ is trivially a projector. 
To verify that the rest are indeed projectors, denote $P(\lambda,a,b,\theta) := \frac{1}{\lambda}\begin{pmatrix} a & e^{-\iota\theta} \\e^{\iota \theta} & b \end{pmatrix}$, and notice that 
\begin{align}\label{eq:projectorver}
P(\lambda,a,b,\theta)^2=\frac{1}{\lambda}\begin{pmatrix} \frac{1+a^2}{\lambda} & \frac{(a+b)e^{-\iota\theta} }{\lambda}\\ \frac{(a+b)e^{\iota \theta} }{\lambda}& \frac{1+b^2}{\lambda} \end{pmatrix}.
\end{align}
Thus $P(\lambda,a,b,\theta)$ is a projector if $1+a^2=\lambda a, 1+b^2=\lambda b$ and $a+b=\lambda$. Taking $(\lambda, a, b)=(2, 1, 1)$, and $\theta=\pi/3$, \eqref{eq:projectorver} implies $P_{u_A^0}^{(A)}(1)$ is a projector. Similarly $(\lambda, a, b)=(4, 2-\sqrt{3}, 2+\sqrt{3})$ shows $P_{u_B^0}^{(B)}(0)$ and $P_{u_B^0}^{(B)}(1)$ are projectors with$\theta=\pi/3$ and $\theta=-\pi/3$, respectively. 

We have the cost of an instance $D=(M_c, N_c, \Pbb, \Uscr_A, \Uscr_B, 2)$ in $\half$-CAC with $\Pbb$ specified in \eqref{eq:1/2cacprior}, under policy $Q$ expressed as
\begin{align*}
&J(Q;D)\\
&=-\sum_{\xi_A, \xi_B}\Pbb(\xi_A, \xi_B, 0) Q(u_A^0, u_B^0|\xi_A, \xi_B)\\
&+2 \Pbb(\xi_A, \xi_B, 1) (Q(u_A^0, u_B^1|\xi_A, \xi_B)+Q(u_A^1, u_B^0|\xi_A, \xi_B) )\\
&=-0.4 (Q(u_A^0, u_B^0|1,1) + Q(u_A^0, u_B^1|0,0)+Q(u_A^1, u_B^0|0,0)\\
&+Q(u_A^0, u_B^1|0,1)+Q(u_A^1, u_B^0|0,1) \\
&+Q(u_A^0, u_B^1|1,0)+Q(u_A^1, u_B^0|1,0)).
\end{align*}
Now from \eqref{eq:quantconditional}, we have $Q(u_A, u_B|\xi_A, \xi_B)=\Tr(\rho_{AB}P_{u_A}^{(A)}(\xi_A)\otimes P_{u_B}^{(B)}(\xi_B))$.  
%It is now straightforward to compute the conditional probabilities $Q$ as given by 
%
We list all the conditional probabilities that appear in our cost below.
\begin{align*}
Q(u_A^0, u_B^0|1,1)&=\frac{\sqrt{3}+2}{8};
Q(u_A^0, u_B^1|0,0)=\frac{\sqrt{3}+2}{16}\\
Q(u_A^1, u_B^0|0,0)&=\frac{3(\sqrt{3}+2)}{16};
Q(u_A^0, u_B^1|0,1)=\frac{\sqrt{3}+2}{16}\\
Q(u_A^1, u_B^0|0,1)&=\frac{3(\sqrt{3}+2)}{16};
Q(u_A^0, u_B^1|1,0)=\frac{1}{4}\\
Q(u_A^1, u_B^0|1,0)&=\frac{\sqrt{3}+2}{8}.
\end{align*}
This strategy thus attains the cost 
\begin{equation}
J(Q;D)=\frac{-7-3\sqrt{3}}{10}\approx -1.22<J_\Lscr^*(D)=-\frac{6}{5},
\end{equation}
and thus finishes our demonstration of the quantum advantage in the $\half$-CAC problem class.

\section{Conclusion}
An exhaustive scan of the introduced superstructure has thus revealed a restriction of the quantum advantage to the CAC and $\half$-CAC classes, which are precisely the ones that admit the coordination dilemma. In addition, these classes do admit the quantum advantage as our numerical demonstration through \cite{deshpande2022quantum} and Section \ref{sec:1/2cacpoc} have revealed. The coordination dilemma is thus central to the advantage offered by the entire set of non-locally correlated strategies that respect absence of communication in the problem. Quantum strategies are indeed, a physically implementable subset of this class.  While our line of analysis has been restricted to a specialised superstructure of binary teams, it hints that the coordination dilemma will remain an intuitive description of the parametric subspaces that admit the quantum advantage in more general problems.

In the successive article of this two part series, we look within the CAC and $\half$-CAC classes, and identify subspaces within them that subsume  the quantum advantage. Our results there characterise the favourable extent of the coordination dilemma for quantum advantage to manifest.

\bibliographystyle{IEEEtran}  
\bibliography{ref}

\begin{thebibliography}{10}
\providecommand{\url}[1]{#1}
\csname url@samestyle\endcsname
\providecommand{\newblock}{\relax}
\providecommand{\bibinfo}[2]{#2}
\providecommand{\BIBentrySTDinterwordspacing}{\spaceskip=0pt\relax}
\providecommand{\BIBentryALTinterwordstretchfactor}{4}
\providecommand{\BIBentryALTinterwordspacing}{\spaceskip=\fontdimen2\font plus
\BIBentryALTinterwordstretchfactor\fontdimen3\font minus
  \fontdimen4\font\relax}
\providecommand{\BIBforeignlanguage}[2]{{%
\expandafter\ifx\csname l@#1\endcsname\relax
\typeout{** WARNING: IEEEtran.bst: No hyphenation pattern has been}%
\typeout{** loaded for the language `#1'. Using the pattern for}%
\typeout{** the default language instead.}%
\else
\language=\csname l@#1\endcsname
\fi
#2}}
\providecommand{\BIBdecl}{\relax}
\BIBdecl

\bibitem{ananthram2007commonrandom}
V.~Ananthram and V.~Borkar, ``Common randomness and distributed control: A
  counterexample,'' \emph{Systems and Control Letters}, 2007.

\bibitem{deshpande2022quantum}
S.~A. Deshpande and A.~A. Kulkarni, ``The quantum advantage in decentralized
  control,'' \emph{https://arxiv.org/abs/2207.12075}, 2022.

\bibitem{yin2020qcryp}
L.~S. e.~a. Yin~J, Li~YH, ``Entanglement-based secure quantum cryptography over
  1,120 kilometres.'' \emph{Nature}, vol. 582, pp. 501--505, 2020.

\bibitem{gisin2007qcomm}
T.~R. Gisin~N, ``Quantum communication,'' \emph{Nature Photon}, vol.~1, pp.
  165--171, 2007.

\bibitem{brunner2014bell}
\BIBentryALTinterwordspacing
N.~Brunner, D.~Cavalcanti, S.~Pironio, V.~Scarani, and S.~Wehner, ``Bell
  nonlocality,'' \emph{Rev. Mod. Phys.}, vol.~86, pp. 419--478, Apr 2014.
  [Online]. Available: \url{https://link.aps.org/doi/10.1103/RevModPhys.86.419}
\BIBentrySTDinterwordspacing

\bibitem{einstein1935epr}
\BIBentryALTinterwordspacing
A.~Einstein, B.~Podolsky, and N.~Rosen, ``Can quantum-mechanical description of
  physical reality be considered complete?'' \emph{Phys. Rev.}, vol.~47, pp.
  777--780, May 1935. [Online]. Available:
  \url{https://link.aps.org/doi/10.1103/PhysRev.47.777}
\BIBentrySTDinterwordspacing

\bibitem{bohr1935epr}
\BIBentryALTinterwordspacing
N.~Bohr, ``Can quantum-mechanical description of physical reality be considered
  complete?'' \emph{Phys. Rev.}, vol.~48, pp. 696--702, Oct 1935. [Online].
  Available: \url{https://link.aps.org/doi/10.1103/PhysRev.48.696}
\BIBentrySTDinterwordspacing

\bibitem{bell1964epr}
\BIBentryALTinterwordspacing
J.~S. Bell, ``On the einstein podolsky rosen paradox,'' \emph{Physics Physique
  Fizika}, vol.~1, pp. 195--200, Nov 1964. [Online]. Available:
  \url{https://link.aps.org/doi/10.1103/PhysicsPhysiqueFizika.1.195}
\BIBentrySTDinterwordspacing

\bibitem{clauser1969chsh}
\BIBentryALTinterwordspacing
J.~F. Clauser, M.~A. Horne, A.~Shimony, and R.~A. Holt, ``Proposed experiment
  to test local hidden-variable theories,'' \emph{Phys. Rev. Lett.}, vol.~23,
  pp. 880--884, Oct 1969. [Online]. Available:
  \url{https://link.aps.org/doi/10.1103/PhysRevLett.23.880}
\BIBentrySTDinterwordspacing

\bibitem{aspect1982violation}
\BIBentryALTinterwordspacing
A.~Aspect, P.~Grangier, and G.~Roger, ``Experimental realization of
  einstein-podolsky-rosen-bohm gedankenexperiment: A new violation of bell's
  inequalities,'' \emph{Phys. Rev. Lett.}, vol.~49, pp. 91--94, Jul 1982.
  [Online]. Available: \url{https://link.aps.org/doi/10.1103/PhysRevLett.49.91}
\BIBentrySTDinterwordspacing

\bibitem{borkar88convex}
V.~S. Borkar, ``A convex analytic approach to markov decision processes,''
  \emph{Probability Theory and Related Fields}, vol.~78, pp. 583--602, Aug.
  1988.

\bibitem{kulkarni2014optimizer}
A.~A. Kulkarni and T.~P. Coleman, ``An optimizer's approach to stochastic
  control problems with nonclassical information structures,'' \emph{IEEE
  Transactions on Automatic Control}, vol.~60, no.~4, pp. 937--949, 2015.

\bibitem{saldi2022geometry}
\BIBentryALTinterwordspacing
N.~Saldi and S.~Y{\"u}ksel, ``{Geometry of information structures, strategic
  measures and associated stochastic control topologies},'' \emph{Probability
  Surveys}, vol.~19, no. none, pp. 450 -- 532, 2022. [Online]. Available:
  \url{https://doi.org/10.1214/20-PS356}
\BIBentrySTDinterwordspacing

\bibitem{matthews2012nonsigcode}
W.~Matthews., ``A linear program for the finite block length converse of
  polyanskiy-poor-verdu' via nonsignaling codes,'' \emph{IEEE Transactions on
  Information Theory}, vol.~58, 2012.

\bibitem{ramanathan2016noquantum}
M.~H. R.~Ramanathan, J.~Tuziemski and P.~Horodecki, ``No quantum realization of
  extremal no-signaling boxes,'' \emph{Phys. Rev. Lett.}, vol. 117, 2016.

\bibitem{werner2001multipartitebell}
M.~W. R.F.~Werner, ``All multipartite bell-correlation inequalities for two
  dichotomic observables per site,'' \emph{Phy. Rev. A}, vol.~64, 2001.

\bibitem{navascues2008hierarchy}
\BIBentryALTinterwordspacing
M.~Navascués, S.~Pironio, and A.~Acín, ``A convergent hierarchy of
  semidefinite programs characterizing the set of quantum correlations,''
  \emph{New Journal of Physics}, vol.~10, no.~7, p. 073013, jul 2008. [Online].
  Available: \url{https://dx.doi.org/10.1088/1367-2630/10/7/073013}
\BIBentrySTDinterwordspacing

\bibitem{popescu1994nonlocquantum}
D.~R. Sandu~Popescu, ``Quantum nonloc ality as an axiom,'' \emph{Foundations of
  Physics}, vol.~24, pp. 379--385, 1994.

\bibitem{barrett2005nonlocal}
\BIBentryALTinterwordspacing
J.~Barrett, N.~Linden, S.~Massar, S.~Pironio, S.~Popescu, and D.~Roberts,
  ``Nonlocal correlations as an information-theoretic resource,'' \emph{Phys.
  Rev. A}, vol.~71, p. 022101, Feb 2005. [Online]. Available:
  \url{https://link.aps.org/doi/10.1103/PhysRevA.71.022101}
\BIBentrySTDinterwordspacing

\bibitem{goh2018geometry}
\BIBentryALTinterwordspacing
K.~T. Goh, J.~Kaniewski, E.~Wolfe, T.~V\'ertesi, X.~Wu, Y.~Cai, Y.-C. Liang,
  and V.~Scarani, ``Geometry of the set of quantum correlations,'' \emph{Phys.
  Rev. A}, vol.~97, p. 022104, Feb 2018. [Online]. Available:
  \url{https://link.aps.org/doi/10.1103/PhysRevA.97.022104}
\BIBentrySTDinterwordspacing

\bibitem{c13pylink}
A.~A.~K. Shashank A~Deshpande, ``{Elimination of 1-3 Classes.}''
  \url{https://tinyurl.com/C13Eliminate}, 2022, [Online; accessed Nov-2022].

\end{thebibliography}

\appendix
\subsection{Proof of propositions in section \ref{sec:equivalences}}\label{app:equiproofs}
\subsubsection{Proof of proposition \ref{prop:transpose}}
We prove (1) $\implies$ (2) for arbitrary $M,N$; replacing $M,N$ with $M\t,N\t$ the reverse implication will follow.     
Suppose $\Cscr(M,N)$ admits a quantum advantage and consider a problem instance $D:=(M, N, \Pbb, \Uscr_A,\Uscr_B, \chi)$ with cost function $ \ell $ 
and a quantum strategy $Q$ $:=(\rho_{AB}$,$\{ P_{u_A}^{(A)}(\xi_A)\}_{u_A,\xi_A}$, $\{ {P}^{(B)}_{u_B}(\xi_B)\}_{u_B, \xi_B})$. 

Now consider an instance $D'$ in the class $\Cscr(M\t, N\t)$ given by  $(M\t, N\t, \Pbb', \Uscr_A',\Uscr_B', \chi')$ where $\chi'  = \chi,  
\Uscr_A'  = \Uscr_B, 
\Uscr_B'  = \Uscr_A $ and 
\begin{align*}
\Pbb'(\xi_A, \xi_B, \xi_W) &\equiv \Pbb(\xi_B, \xi_A, \xi_W). 
\end{align*}
It follows that the cost $\ell'$ of $D'$ satisfies $\ell'(u_A',u_B',\xi_W) \equiv \ell(u_B',u_A',\xi_W).$
Consider a strategy $Q':=(\rho_{AB}\t$, $\{{P}'^{(A)}_{u_A'}(\xi_A)\}_{u_A', \xi_A},$ $\{{P}'^{(B)}_{u_B'}(\xi_B)\}_{u_B', \xi_B} )$
where
\begin{align*}
{P}'^{(A)}_{u_A'}(\xi_A)&=P^{(B)}_{u_A'}(\xi_A),\quad  \forall u_A' \in \Uscr_A', \xi_A \in \Xi_A, \\
{P}'^{(B)}_{u_B'}(\xi_B)&=P^{(A)}_{u_B'}(\xi_B),\quad \forall u_B' \in \Uscr_B', \xi_B \in \Xi_B.
\end{align*}
Now by properties of the trace, for all $\xi_A \in \Xi_A$ ,$\xi_B \in \Xi_B$, $u_A' \in \Uscr_A'$, $u_B'\in \Uscr_B'$, 
\[\Tr\rho_{AB}\t {P}'^{(A)}_{u_A'}(\xi_A){P}'^{(B)}_{u_B'}(\xi_B) = \Tr\rho_{AB} {P}^{(A)}_{u_B'}(\xi_B) {P}^{(B)}_{u_A'}(\xi_A)\]
Thus $Q'$ satisfies for all $u_A' \in \Uscr_A', u_B'\in \Uscr_B'$,  \begin{equation}
Q'(u_A',u_B'|\xi_A,\xi_B) = Q(u_B',u_A'|\xi_B,\xi_A),\label{eq:tpolicycorrespondence} 
\end{equation} 
and hence from \eqref{eq:qcost}
\begin{align*}
J(Q';D')
&=	\sum_{\xi_A, \xi_B, \xi_W}\Pbb'(\xi_A, \xi_B, \xi_W)\sum_{ u_A', u_B'} \ell'(u_A', u_B', \xi_W)\\
&\qquad \qquad \times\Tr\left(\rho_{AB}P_{u_B'}^{(A)}(\xi_B)P_{u_A'}^{(B)}(\xi_A)\right) \\
&= \sum_{\xi_A, \xi_B, \xi_W}\Pbb(\xi_B, \xi_A, \xi_W)\sum_{ u_A', u_B'} \ell(u_B', u_A', \xi_W)\\
&\qquad \qquad\times  \Tr\left(\rho_{AB}P_{u_B'}^{(A)}(\xi_B)P_{u_A'}^{(B)}(\xi_A)\right) \\
& = J(Q;D),\label{eq:tcostcorrespondence}
\end{align*}
where we have used that $u_A'\in \Uscr_A'=\Uscr_B$ and $u_B'\in \Uscr_B'=\Uscr_A.$  
Similarly for any deterministic strategy $\gamma\in\Gamma$ consider a $\gamma'\in\Gamma'$ for the instance $D'$ such that $\gamma'_A(\xi_A)=\gamma_B(\xi_A)$ and $\gamma'_B(\xi_B)=\gamma_A(\xi_A)$  so that
\begin{align*}
J(\pi_{\gamma'};D')
%=\sum_{\xi_A, \xi_B, \xi_W}\Pbb'(\xi_A, \xi_B, \xi_W)\ell'(\gamma'_A(\xi_A), \gamma'_B(\xi_B), \xi_W)\\
%&=\sum_{\xi_A, \xi_B, \xi_W}\Pbb(\xi_A, \xi_B, \xi_W)\ell(\gamma_A(\xi_A), \gamma_B(\xi_B), \xi_W)\\
&=J(\pi_{\gamma};D),
\end{align*}
and hence $J^*_\Lscr(D)=J^*_\Lscr(D').$ 
Thus if $\exists Q\in \Qscr $ such that $  J(Q; D)<J^*_\Lscr(D)$ then with $Q',D'$ as above, we have $J(Q'; D')< J^*_\Lscr( D'),$ whereby $\Cscr(M\t,N\t)$ admits a quantum advantage. 
We have established the proposition.

\subsubsection{Proof of proposition \ref{prop:permute}}
($(1)\Leftrightarrow (3)$) We will show that $(3)$ implies $(1)$. But since $\Rsf^2=\Ibf$, applying the same to result with $M,N$ replaced by $\Rsf M,\Rsf N$, we can conclude that $(1) \implies (3)$.
 For $D= (M, N, \Pbb, \Uscr_A, \Uscr_B, \chi) \in \Cscr(M, N)$, consider a $D'=(\mathsf{R}M, \mathsf{R}N, \Pbb, \Uscr_A', \Uscr_B, \chi) \in \Cscr(\mathsf{R}M, \mathsf{R}N)$ and let its cost function be denoted $\ell'$. We think of $\Uscr_A $ as a column vector $(u_A^0,u_A^1)$ and let $\Uscr_A' = \Rsf \Uscr_A$. 
Then for any strategy $Q \in \Pscr(\Uscr|\Xi)$, we have, 
\begin{align*}
	&J(Q; D')\\
 &=\sum_{\xi, u_B}\Pbb(\xi)\sum_{u_A'\in \Uscr_A'}\ell'(u_A^{\prime}, u_B, \xi_W)Q(u_A^{\prime}, u_B|\xi_A, \xi_B)\\
&=\sum_{\xi, u_B}\Pbb(\xi)\sum_{u_A \in \Uscr_A}\ell(u_A, u_B, \xi_W)Q(u_A, u_B|\xi_A, \xi_B) \\
&=J(Q; D),
\end{align*}
where $\xi = (\xi_A,\xi_B,\xi_W).$ It is easy to also see that $J_\Pi^*(D)=J_\Pi^*(D').$ Thus, 
$$J(Q, D)-J_\Pi^*(D)<0\implies J(Q, D')-J_\Pi^*(D)<0$$
and we have shown that $(3) \implies (1)$, and thus $(1) \Leftrightarrow (3).$

The equivalence between $(2)$ and $(3)$ can be shown in a similar manner, thereby completing the proof.

\subsubsection{Proof of proposition \ref{prop:exchange}}
It suffices to show $(1) \implies (2)$ for arbitrary $M,N$, following which the reverse claim will follow. 
Let $D:=(M, N, \Pbb, \Uscr_A, \Uscr_B, \chi)$ and $Q$ $:=(\rho_{AB}$,$\{ P_{u_A}^{(A)}(\xi_A)\}_{u_A,\xi_A}$,$\{ {P}^{(B)}_{u_B}(\xi_B)\}_{u_B, \xi_B})\in \Qscr$. 
%\begin{multline}
%J(Q; D)=\sum_{\xi}\Pbb(\xi_A, \xi_B, \xi_W)\times \\
%\sum_{ u} \ell(u_A, u_B, \xi_W)\Tr\left(\rho_{AB}P_{u_A}^{(A)}(\xi_A)P_{u_B}^{(B)}(\xi_B)\right),
%\end{multline}
Consider a problem $D':=(N, M, \Pbb', \Uscr_A, \Uscr_B,1/\chi) \in \Cscr(N,M)$ where
\begin{equation*}
\Pbb'(\xi_A, \xi_B, 0)=\Pbb(\xi_A, \xi_B, 1), \quad \Pbb'(\xi_A, \xi_B, 1)=\Pbb(\xi_A, \xi_B, 0)
\end{equation*}
Notice that the cost function $\ell'(.)$ of the instance $D'$ is related to $\ell(.)$ of $D$ as 
\begin{align*}
\ell'(\cdot, \cdot, 0)&\equiv \frac{1}{\chi}\ell(\cdot, \cdot, 1),\quad 
\ell'(\cdot,\cdot, 1)\equiv \frac{1}{\chi}\ell(\cdot,\cdot, 0)
\end{align*}
Thus, denoting $\xi=(\xi_A,\xi_B,\xi_W),$ we have for any $Q\in \Pscr(\Uscr|\Xi)$
\begin{align}
J(Q;D')&=\sum_{u, \xi}\Pbb'(\xi)\ell'(u_A, u_B, \xi_W)Q(u_A, u_B|\xi_A, \xi_B)\nonumber\\
&=\frac{1}{\chi}\sum_{u, \xi} \Pbb(\xi)\ell(u_A, u_B, \xi_W)Q(u_A, u_B|\xi_A, \xi_B)\nonumber\\
&=\frac{1}{\chi} J(Q;D) \label{eq:Eequivalence}
\end{align}
In particular taking $Q \in \Pi$, we get that $J^*_\Lscr(D) = \frac{1}{\chi}J^*_\Lscr(D')$. Thus if $\exists Q\in \Qscr$ such that  $J(Q;D)-J^*_\Lscr(D)<0$, then for this $Q$, we have
$J(Q; D')-J^*_\Lscr(D')=\frac{1}{\chi}(J(Q;D)-J^*_\Lscr(D))<0.$
This establishes the proposition.

\subsection{Optimal Deterministic Policy for the $\half$-CAC instance in Section \ref{sec:1/2cacpoc}}

\begin{lemma}\label{lem:1/2cacdemodet}
For the instance $D=(M_c, N_c, \Pbb, \Uscr_A, \Uscr_B, \chi)\in \half$-CAC with $\Pbb$ given by \eqref{eq:1/2cacprior} and $\chi=2$, specified in Section \ref{sec:1/2cacpoc}, an optimal deterministic policy and the corresponding cost is given by
$$\gamma^*_A\equiv u_A^0, \gamma_B^*\equiv u_B^1, \quad J(\pi_{\gamma^*};D)=-6/5.$$
\end{lemma}
\begin{proof}
We show this by directly enumerating all sixteen policies and their cost in the table below. 
\begin{center}
    \begin{tabular}{c|c|c|c|c}
         $\gamma_A(0)$ & $\gamma_A(1)$ & $\gamma_B(0)$ & $\gamma_B(1)$ & $J(\pi_{\gamma}; D)$  \\
        \hline
          $u_A^0$&$u_A^0$&$u_B^0$&$u_B^0$& $-2/5$\\
          $u_A^1$&$u_A^0$&$u_B^0$&$u_B^0$& $-6/5$\\
          $u_A^0$&$u_A^1$&$u_B^0$&$u_B^0$& $-2/5$\\
          $u_A^1$&$u_A^1$&$u_B^0$&$u_B^0$& $-6/5$\\
          $u_A^0$&$u_A^0$&$u_B^1$&$u_B^0$& $-6/5$\\
          $u_A^1$&$u_A^0$&$u_B^1$&$u_B^0$& $-6/5$\\
          $u_A^0$&$u_A^1$&$u_B^1$&$u_B^0$& $-2/5$\\
          $u_A^1$&$u_A^1$&$u_B^1$&$u_B^0$& $-2/5$\\
          $u_A^0$&$u_A^0$&$u_B^0$&$u_B^1$& $-2/5$\\
          $u_A^0$&$u_A^1$&$u_B^0$&$u_B^1$& $-4/5$\\
          $u_A^1$&$u_A^0$&$u_B^0$&$u_B^1$& $-2/5$\\
          $u_A^1$&$u_A^1$&$u_B^0$&$u_B^1$& $-4/5$\\
          $\boxed{u_A^0}$&$\boxed{u_A^0}$&$\boxed{u_B^1}$&$\boxed{u_B^1}$& $\boxed{-6/5}$\\
          $u_A^0$&$u_A^1$&$u_B^1$&$u_B^1$& $-4/5$\\
          $u_A^1$&$u_A^0$&$u_B^1$&$u_B^1$& $-2/5$\\
          $u_A^1$&$u_A^1$&$u_B^1$&$u_B^1$& $0$\\    
         \hline
    \end{tabular}
\end{center}
It is evident that the boxed policy is an optimal policy.
\end{proof}

\begin{comment}
\begin{multline}
((\sim(((\sim\beta\cdot\sim\delta)\vee(\beta\cdot\alpha\cdot\delta)\cdot\xi_A)\oplus(\delta\vee(\sim\delta\cdot\alpha\cdot\beta))))\cdot(\sim((\sim\alpha\cdot\sim\delta\cdot\beta\cdot\xi_B)\oplus(\alpha\cdot(\beta\oplus\delta))\vee(\sim\alpha\cdot\delta))))\\
\vee ((\sim((\xi_A)\oplus(\sim\beta\vee(\beta\cdot\sim\alpha\cdot\sim\delta))))\cdot(\sim((\xi_B)\oplus(\sim\alpha\cdot(\beta\vee\delta))\vee(\alpha\cdot(\sim\beta\oplus\delta)))))
\end{multline}
\end{comment}

\begin{biography}[{\includegraphics[width=1in,height=1in]{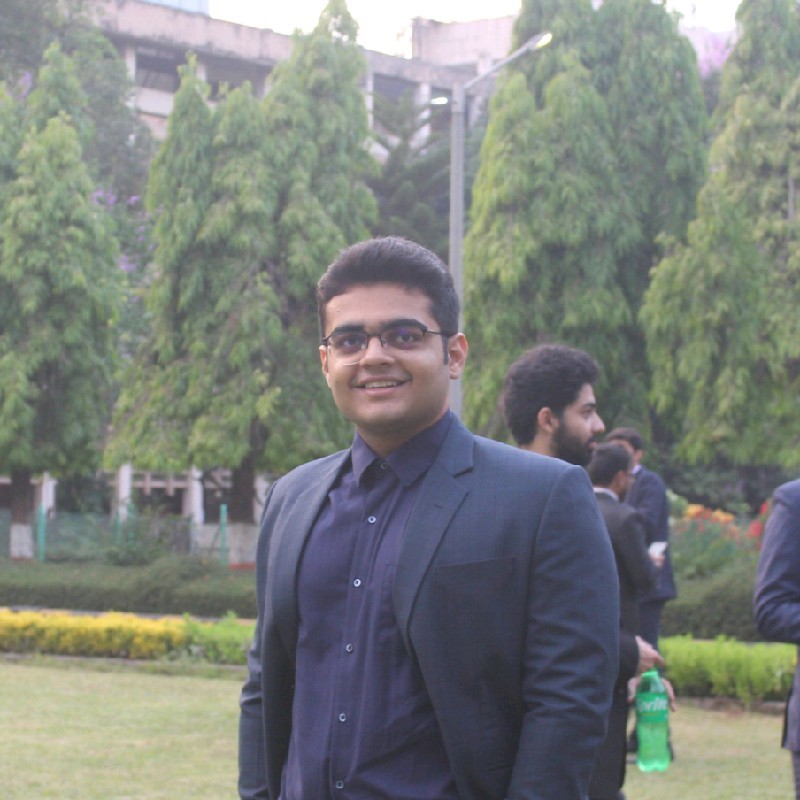}}]{Shashank Aniruddha Deshpande}
 was born in Buldhana, India, in 2000. He is an undergraduate senior in the Department of Physics, and, the Department of Systems and Control Engineering at IIT Bombay, India. His research interests are in the control and optimization of stochastic and networked systems.
\end{biography}

\begin{biography}[{\includegraphics[width=1in, keepaspectratio]{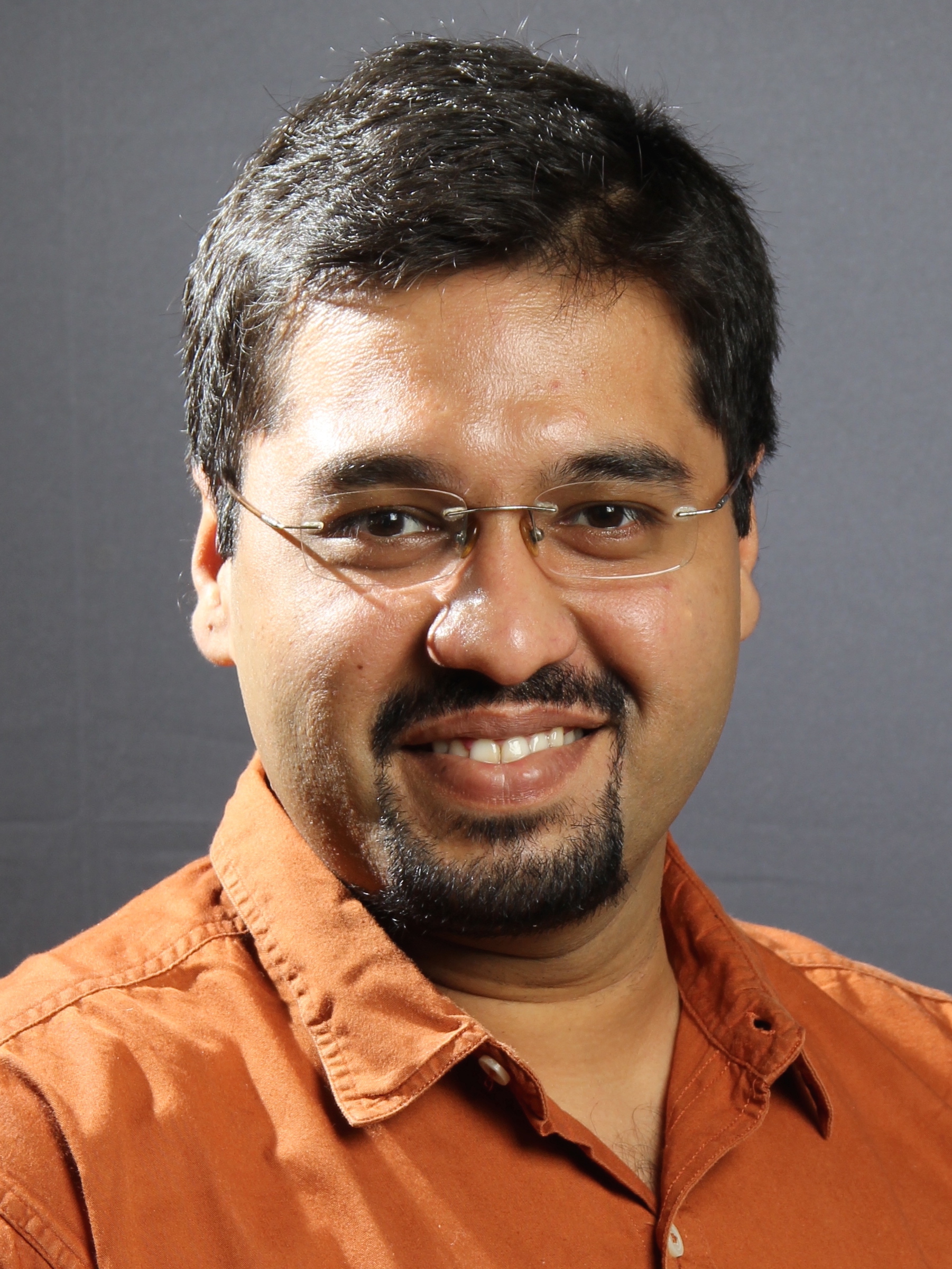}}]{Ankur A. Kulkarni}
	is the Kelkar Family Chair Associate Professor with the Systems and Control Engineering group at the Indian Institute of Technology Bombay (IITB). He received his B.Tech. from IITB in 2006, followed by M.S. in 2008 and Ph.D. in 2010, both from the University of Illinois at Urbana-Champaign (UIUC). From 2010-2012 he was a post-doctoral researcher at the Coordinated Science Laboratory at UIUC. He was an Associate (from 2015--2018) of the Indian Academy of Sciences, a recipient of the INSPIRE Faculty Award of the Dept of Science and Technology Govt of India in 2013. He has won
	several best paper awards at conferences, the Excellence in Teaching Award in 2018 at IITB, and the William A. Chittenden Award in 2008 at UIUC.
	He has been a consultant to the Securities and Exchange Board of India (SEBI), the HDFC Life Insurance Company, Kotak Mahindra Bank Ltd and Bank of Baroda. He presently serves on the IT-Project Advisory Board of SEBI, as Research Advisor to the Tata Consultancy Services, and as Program Chair of the Indian Control Conference. He has been a visitor to MIT in USA, University of Cambridge in UK, NUS in Singapore, IISc in Bangalore and KTH in Sweden.
\end{biography}

\end{document}